\documentclass[journal,9pt]{IEEEtran}
\usepackage[colorlinks,linkcolor=black,anchorcolor=black,linktocpage=true,urlcolor=black,citecolor=black]{hyperref}

\usepackage{stfloats}
\usepackage{booktabs}
\usepackage{algorithmic}
\usepackage{tikz}
\usetikzlibrary{arrows}
\usepackage{array}
\usepackage{graphicx}
\usepackage{url}
\usepackage{cite}
\usepackage{amsmath,amsthm,amsfonts,amssymb,mathrsfs,bm}
\allowdisplaybreaks[4]
\usepackage{tabularx,array,enumerate,pifont,tikz,epstopdf}

\usetikzlibrary{shapes}

\makeatletter
\newcommand{\biggg}{\bBigg@{3}}

\newcommand{\Biggg}{\bBigg@{3.5}}
\def\Bigggl{\mathopen\Biggg}

\newcommand{\bigggg}{\bBigg@{5}}

\newcommand{\Bigggg}{\bBigg@{6.5}}

\makeatother

\newtheorem{theorem}{Theorem}
\newtheorem{lemma}{Lemma}

\newtheorem{remark}{Remark}

\newtheorem{assumption}{Assumption}

\newtheorem{problem}{Problem}

\everydisplay{
	\abovedisplayskip=.25\baselineskip plus.2ex minus.2ex\abovedisplayshortskip=-0.33\baselineskip plus.4ex minus.4ex
	\belowdisplayskip=.25\baselineskip plus.2ex minus.2ex\belowdisplayshortskip=.25\baselineskip plus.2ex minus.2ex
}

\begin{document}

\title{\hspace*{-2.8mm} \LARGE \bf Distributed Optimal Output Consensus of Uncertain Nonlinear Multi-Agent Systems over Unbalanced Directed Networks via Output Feedback}

\author{Jin~Zhang,
        Lu~Liu,~\IEEEmembership{Senior Member, IEEE,}
        Xinghu~Wang,
        and~Haibo~Ji% <-this % stops a space
\thanks{J. Zhang is with the Department of Automation, University of Science and Technology of China, Hefei 230027, China, and also with the Department of Biomedical Engineering, City University of Hong Kong, Hong Kong (e-mail: zj55555@mail.ustc.edu.cn).
	
	L. Liu is with the Department of Biomedical Engineering, City University of Hong Kong, Hong Kong (e-mail: luliu45@cityu.edu.hk).
	
	X. Wang and H. Ji are with the Department of Automation, University of Science and Technology of China, Hefei 230027, China (e-mail: xinghuw@ustc.edu.cn; jihb@ustc.edu.cn).}}

%\markboth{Journal of \LaTeX\ Class Files,~Vol.~14, No.~8, August~2020}%

\maketitle

\begin{abstract}
	In this note, a novel observer-based output feedback control approach is proposed to address the distributed optimal output consensus problem of uncertain nonlinear multi-agent systems in the normal form over unbalanced directed graphs. The main challenges of the concerned problem lie in unbalanced directed graphs and nonlinearities of multi-agent systems with their agent states not available for feedback control. Based on a two-layer controller structure, a distributed optimal coordinator is first designed to convert the considered problem into a reference-tracking problem. Then a decentralized output feedback controller is developed to stabilize the resulting augmented system. A high-gain observer is exploited in controller design to estimate the agent states in the presence of uncertainties and disturbances so that the proposed controller relies only on agent outputs. The semi-global convergence of the agent outputs toward the optimal solution that minimizes the sum of all local cost functions is proved under standard assumptions. A key feature of the obtained results is that the nonlinear agents under consideration are only required to be locally Lipschitz and possess globally asymptotically stable and locally exponentially stable zero dynamics.
\end{abstract}


\IEEEpeerreviewmaketitle

\section{Introduction}
\IEEEPARstart{I}{n} the past decade, distributed optimization has made remarkable advancement and is widely applied in a broad range of fields, such as machine learning, power systems and so on \cite{Yang2019survey,Molzahn2017survey}. Distributed optimization aims at achieving an optimal consensus, which minimizes the sum of all local cost functions attached to individual agents, in a distributed manner. A considerable volume of works on solving distributed optimization problems have been reported, in both discrete-time \cite{Nedic2009distributed,Nedic2014distributed,Xi2017dextra} and continuous-time settings \cite{Gharesifard2013distributed,Kia2015distributed}.
Many engineering tasks can be formulated as a distributed optimal output consensus (DOOC) problem of multi-agent systems with general agent dynamics, such as economic dispatch in power systems \cite{stegink2016Unifying} and source seeking in multi-robot systems \cite{zhang2011extremum}. So far, much effort has been devoted to the DOOC problem of multi-agent systems with double integrators agent dynamics \cite{zhang2017distributed} and high-order linear agent dynamics \cite{li2019distributed,zhang2020exponential}. 

More recently, DOOC problems of more general multi-agent systems with nonlinear agent dynamics have been investigated \cite{Wang2015distributed,Tang2018distributed,zhang2021optimal,li2020distributed,tang2020optimal}. The DOOC problem of multi-agent systems with nonlinear agent dynamics in output feedback form over undirected networks is addressed in \cite{Wang2015distributed}. Later, the authors in \cite{Tang2018distributed} develop an adaptive controller to tackle the difficulty brought by unknown nonlinear agent dynamics though still over undirected networks. It is worth noting that the controller developed in \cite{Tang2018distributed} is based on a two-layer framework, which consists of an optimal coordinator generating the optimal solution and a decentralized output feedback controller driving each agent to track its individual optimal coordinator. Then, in our preliminary work \cite{zhang2021optimal}, the two-layer framework is extended to solve the DOOC problem of disturbed second-order nonlinear systems but over unbalanced directed networks. The DOOC problem of more general nonlinear multi-agent systems in the normal form over undirected or balanced directed networks are addressed in \cite{tang2020optimal,li2020distributed}, where the inverse dynamics of agents are unfortunately required to be input-to-state stable (ISS). Moreover, the controllers developed in \cite{tang2020optimal,li2020distributed} are based on agent state feedback, which may not be always available for many practical multi-agent systems.

It is of both theoretical and practical significance to study unbalanced directed graphs as the information exchange between agents may be unidirectional due to limited bandwidth or other constraints. Inspired by the graph balancing technique in \cite{mei2015distributed}, a distributed continuous-time control strategy is proposed in \cite{Li2017distributed}, where the left eigenvector corresponding to the zero eigenvalue of the Laplacian matrix is utilized to tackle weight-unbalanced directed networks. However, this technique cannot be adopted when the left eigenvector is not known\textit{ a priori}. To remove the explicit dependency on the left eigenvector corresponding to the zero eigenvalue of the Laplacian matrix, the authors in \cite{Zhu2018continuous} propose a novel distributed algorithm with its gradient term being divided by an auxiliary variable. Nevertheless, one common limitation of the above-mentioned works is that only single integrator agent dynamics are considered.

Motivated by the above observation, this note investigates the DOOC problem of uncertain nonlinear multi-agent systems with nonlinear agent dynamics in the normal form over unbalanced directed networks via output feedback. To address main challenges brought by unbalanced directed graphs and uncertain nonlinear dynamics, the concerned DOOC problem is first converted to a reference-tracking problem by designing a distributed optimal coordinator, and the resulting augmented system is then stabilized by an observer-based output feedback controller. The main contributions of this work are summarized as follows.

1) In contrast to linear agent dynamics or nonlinear agent dynamics in output feedback form considered in \cite{zhang2020exponential,Wang2015distributed,Tang2018distributed,Zhu2018continuous,Xie2020suboptimal,Li2017distributed,zhang2021optimal}, the nonlinear agent dynamics in the normal form are more general and include the above-mentioned agent dynamics as special cases \cite{khalil2002nonlinear}. Furthermore, in virtue of a high-gain observer, only agent outputs are needed for controller design in this note, which greatly enhances its applicability in practice.

2) Unlike the works \cite{Gharesifard2013distributed,Kia2015distributed,Wang2015distributed,Tang2018distributed,tang2020optimal,li2020distributed} that study undirected or balanced directed graphs, this work considers more general and also more challenging weight-unbalanced directed networks. It is shown that the proposed controller is able to deal with the imbalance resulting from general directed networks and drive the agent outputs to the optimal solution.

3) Compared with the existing works \cite{tang2020optimal,li2020distributed}, where the inverse dynamics of agents are assumed to be ISS, this work only requires the zero dynamics of agents to be globally asymptotically stable and locally exponentially stable. It is shown by semi-global stability analysis that the convergence can be achieved via a linear controller instead of a nonlinear one even under this less stringent assumption. It is worth pointing out that a linear controller is advantageous in both theoretical design and practical implementation, and thus has significant engineering implications.

The rest of this note is organized as follows. Some preliminaries and the problem formulation are given in Section \ref{section preliminaries}. The main results of this note and a simulation example are provided in Section \ref{section main results} and Section \ref{section simulation results}, respectively. The conclusion and future works are stated in Section \ref{section conclusion}.

\textit{Notations}: $\|\cdot\|$ denotes the Euclidean norm of vectors or induced 2-norm of matrices. $x^{\mathrm{T}}$ and $A^{\mathrm{T}}$ refer to the transpose of vector $ x $ and matrix $ A $, respectively. $ \operatorname{col}\left(x_{1}, x_{2}, \ldots, x_{n}\right) $ represents a column vector with $x_{1}, x_{2}, \ldots, x_{n}$ being its elements. $\operatorname{diag}\left(x_{1}, x_{2}, \ldots, x_{n}\right)$ represents a diagonal matrix with $x_{1}, x_{2}, \ldots, x_{n}$ being its diagonal elements. For a differentiable function $f: \mathbb{R}^{n} \rightarrow \mathbb{R}$, $\nabla f$ denotes its gradient.  $\bar{Q}_{R}^{n} \triangleq\{x=\operatorname{col}(x_{1}, \ldots, x_{n}) \in \mathbb{R}^{n}:|x_{i}| \leq R, i=1, \ldots, n\}$ and $\bar{\varOmega}_{c}(V(x)) \triangleq \{x \in \mathbb{R}^{n}: V(x) \leq c\}$ are compact sets, where $V: \mathbb{R}^{n} \rightarrow \mathbb{R}$ is a positive definite function.

\section{Preliminaries and Problem Formulation}\label{section preliminaries}
In this section, we present some preliminaries on graph theory and convex analysis, and then formulate the problem under consideration. 

\subsection{Graph Theory} \label{subsection_graph_theory}
A directed graph, a digraph in short, can be described by a triplet $\mathcal{G}=(\mathcal{V}, \mathcal{E}, \mathcal{A})$, which is composed of a set $\mathcal{V}=\{1, \ldots, N\}$ of nodes, a collection $\mathcal{E} \subseteq \mathcal{V} \times \mathcal{V}$ of ordered pairs of nodes, called edges, and an adjacency matrix $\mathcal{A}$. 
For $ i,j\in\mathcal{V} $, the ordered pair $(j, i) \in \mathcal{E}$ denotes an edge from $ j $ to $ i $, in this case, node $ j $ is called an in-neighbor of node $ i $. $ \mathcal{N}_{i} $ denotes the set containing all the in-neighbors of agent $ i $. A directed path is an ordered sequence of nodes in which any pair of consecutive nodes is a directed edge. A self-loop is an edge from a node to itself. A digraph is said to be strongly connected if, for any node, there exists a directed path from any other node to itself. 
The adjacency matrix $\mathcal{A}=\left[a_{i j}\right] \in \mathbb{R}^{N \times N}$ is defined as $a_{i j}>0$ if $(j, i) \in \mathcal{E},$ otherwise $a_{i j}=0 $. Besides, $a_{i i}=0$ for all $ i $ since there is no self-loop. Moreover, the associated Laplacian matrix $\mathcal{L}=\left[l_{i j}\right] \in \mathbb{R}^{N \times N}$ is defined as $l_{i i}=\sum_{j=1}^{N} a_{i j}$, and $l_{i j}=-a_{i j}$ for $i \neq j $. A digraph $\mathcal{G}$ is weight balanced if and only if $\mathbf{1}_{N}^{\mathrm{T}} \mathcal{L}=\mathbf{0}_{N}^{\mathrm{T}}$. One may refer to \cite{Bullo2019lectures} for more details on graph theory.

\subsection{Convex Analysis} \label{subsection_convex_analysis}
A differentiable function $c: \mathbb{R}^{n} \rightarrow \mathbb{R}$ is said to be $\varpi $-strongly convex on $\mathbb{R}^{n}$ if there exists a constant $ \varpi>0 $ such that $(x-y)^{\operatorname{T}}(\nabla c(x)-\nabla c(y)) \geq \varpi\|x-y\|^{2}$ for all $x, y \in \mathbb{R}^{n} $. A function $g: \mathbb{R}^{n} \rightarrow \mathbb{R}^{n}$ is said to be globally Lipschitz on $\mathbb{R}^{n}$ if there exists a constant $ l>0 $ such that $\|g(x)-g(y)\| \leq l\|x-y\|$ for all $x, y \in \mathbb{R}^{n}$. One may refer to \cite{Bertsekas2009convex} for more details.

\subsection{Problem Formulation} \label{section problem formulation}

Consider a multi-agent system composed of $ N $ agents over a weight-unbalanced directed network. The dynamics of each agent are in the normal form described as follows,
\begin{equation}\label{dynamics}
	\begin{aligned} 
		\dot{z}_{i}=&f_{i 0}\left(z_{i}, x_{i 1}, v, w\right), \\
		\dot{x}_{i k}\!=&x_{i(k+1)}, \quad k=1,2, \ldots n_{i}-1,\\
		\dot{x}_{i n_{i}}\!=&f_{i 1}\left(z_{i}, x_{i}, v, w\right)+b_{i}(w) u_{i}, \\
		y_{i}\!=&x_{i 1}, \quad i=1,2, \ldots, N,
	\end{aligned}
\end{equation}
where $ \operatorname{col}(z_{i}, x_{i}) $ is the state with $ z_{i}\in\mathbb{R}^{n_{z_{i}}} $ and $ x_{i}=\operatorname{col}(x_{i1},x_{i2},\ldots,x_{in_{i}}) \in\mathbb{R}^{n_{i}}$, $u_{i} \in \mathbb{R}$ and $ y_{i} \in \mathbb{R}$ are the control input and measurement output, respectively. $ w \in\mathbb{R}^{n_{w}} $ represents the parametric uncertainty. $v \in \mathbb{R}^{n_{v}}$ is an exogenous signal representing the disturbance generated by an autonomous exosystem,
\begin{equation} \label{exosystem}
	\dot{v}=Sv,
\end{equation} 
where $ S\in\mathbb{R}^{n_{v}\times n_{v}} $. For $ i=1,2, \ldots, N $, it is assumed that $ f_{i 0} $, $f_{i 1}$ and $ b_{i} $ are sufficiently smooth functions satisfying $f_{i0}(0,0,0, w)=0$, $f_{i1}(0,0,0, w)=0$, and $ b_{i}(w)>0 $ for all $ w\in\mathbb{R}^{n_{w}} $.

\begin{remark}
	The DOOC problem of multi-agent systems with nonlinear agent dynamics in the normal form over undirected or balanced directed networks has been investigated in \cite{tang2020optimal,li2020distributed}. However, the following significant differences should be noted. Firstly, this note considers the same DOOC problem of uncertain nonlinear multi-agent systems with nonlinear agent dynamics in the same normal form but over unbalanced directed networks. Secondly, the inverse dynamics of agents are assumed to be ISS in \cite{tang2020optimal,li2020distributed}, while this work only requires the zero dynamics of agents to be globally asymptotically stable as well as locally exponentially stable. Thirdly, the nonlinear functions in \cite{li2020distributed} are required to be globally Lipschitz, while it is not required in this work. Fourthly, external disturbances generated by the exosystem (\ref{exosystem}) are considered in this work, while they are not considered in \cite{tang2020optimal,li2020distributed}. It is also worth pointing out that the exosystem (\ref{exosystem}) can produce a large class of external signals, such as sinusoidal, step and ramp type signals.
\end{remark}

In addition, it is assumed that each agent $ i $ is assigned with an individual local cost function $ c_{i}(s): \mathbb{R}\to \mathbb{R} $ with its decision variable $ s\in \mathbb{R} $. It should be noted that the local cost function $ c_{i}(\cdot) $ is only accessable to agent $ i $. Define the global cost function and its corresponding optimal solution as $ c(s)=\sum_{i=1}^{N}c_{i}(s) $ and $ s^{\star}\in \mathbb{R} $, respectively. To seek the optimal solution in a distributed manner, the controller design for each agent is only allowed to exploit avaliable information from its in-neighbors and itself. Specifically, the controller is expected to take the following form,
\begin{equation} \label{controller form_origin system}
	\begin{aligned}
	u_{i} = \kappa_{i1}\big(\nabla f_{i}, y_{i}, \upsilon_{i}\big), \quad
	\dot{\upsilon}_{i} = \kappa_{i2}\big(\nabla f_{i}, y_{i}, \upsilon_{j}, j\in \bar{\mathcal{N}}_{i}\big),
	\end{aligned}
\end{equation}
where $ \kappa_{i1} $ and $ \kappa_{i2} $ are functions vanishing at the origin, $ \bar{\mathcal{N}}_{i} = \mathcal{N}_{i} \cup \{i\} $ is a set containing the in-neighbors of agent $ i $ and itself, $ \upsilon_{i}\in \mathbb{R}^{n_{\upsilon_{i}}} $ is a state of the dynamic controller with its dimention $ n_{\upsilon_{i}} $ to be specified later.
Let $ x_{c}=\operatorname{col}(z_{1},x_{1},\upsilon_{1},\ldots,z_{N},x_{N},\upsilon_{N}) $ and $ n_{c}=\sum_{i=1}^{N}(n_{z_{i}}+n_{i}+n_{\upsilon_{i}}) $.
Then the following problem can be formulated.
\begin{problem} \label{problem}
	Consider the multi-agent system (\ref{dynamics}) and the exosystem (\ref{exosystem}) under the directed graph $ \mathcal{G} $ with local cost functions $ c_{i}(\cdot), i=1,2,\ldots,N $. Given any constant $ R>0 $ and any nonempty compact set $ \mathbb{V}_{0}\times\mathbb{W}\subseteq \mathbb{R}^{n_{v}\times n_{w}} $ containing the origin, design a distributed dynamic output feedback controller of the form (\ref{controller form_origin system}) such that, for any initial $ x_{c}(0)\in\bar{Q}_{R}^{n_{c}} $ and $ \operatorname{col}(v(0),w)\in \mathbb{V}_{0}\times \mathbb{W} $, the trajectories of the closed-loop system consisting of (\ref{dynamics}) and (\ref{controller form_origin system}) exist and are bounded for all $ t\geq 0 $, moreover, all the outputs $ y_{i}, i=1,2,\ldots,N $ converge to the optimal value $ s^{\star} $ that minimizes the global cost function $ c(s) = \sum_{i=1}^{n}c_{i}(s) $ as time goes to infinity.
\end{problem}

The following assumptions are needed for solving Problem \ref{problem}.

\begin{assumption} \label{assumption_cost functions}
	For $ i=1,2,\ldots,N $, the local cost function $ c_{i} $ is continuously differentiable and $ \varpi_{i} $-strongly convex, and $\nabla c_{i}$ is globally Lipschitz on $ \mathbb{R} $ with constant $ l_{i} $.
	\label{cost function assumption}
\end{assumption}

\begin{assumption} \label{graph assumption}
	The directed graph $ \mathcal{G} $ is strongly connected.
\end{assumption}

\begin{remark} \label{remark_cost function_2}
	Under Assumption \ref{assumption_cost functions}, the existence and uniqueness of the optimal solution $ s^{\star}\in \mathbb{R} $ can be guaranteed. Under Assumption \ref{graph assumption}, one is able to obtain a weight-balanced digraph by virtue of the graph balancing technique as in \cite{mei2015distributed,Bullo2019lectures}. Assumptions \ref{assumption_cost functions} and \ref{graph assumption} are commonly used in solving the distributed optimization problem over directed graphs, see, for example \cite{Kia2015distributed,Zhu2018continuous}.
\end{remark}

\begin{assumption}\label{assumption_zero_dynamics}
	For $ i=1,2,\ldots,N $, there exists a sufficiently smooth function $ z_{i}^{\star}(s,v,w) $ with $ z_{i}^{\star}(0,0,w)=0 $ such that, for any $ \operatorname{col}(v,w)\in \mathbb{R}^{n_{v}\times n_{w}} $ and $ s\in\mathbb{R} $,
	\begin{equation*}
		\frac{\partial z_{i}^{\star}(s,v,w)}{\partial v}Sv=f_{i0}(z_{i}^{\star}(s,v,w),s,v,w).
	\end{equation*}
\end{assumption}

\begin{assumption} \label{assumption_exosysem}
	The exosystem is neutrally stable, i.e., all the eigenvalues of $ S $ are semi-simple with zero real parts.
\end{assumption}

\begin{remark} \label{remark_d}
	Assumption \ref{assumption_zero_dynamics} is commonly used in solving the cooperative output regulation problem of multi-agent systems with nonlinear agent dynamics in the normal form \cite{su2014cooperative,su2019semi}.
	Under Assumption \ref{assumption_exosysem}, given any compact set $ \mathbb{V}_{0} $, it can be shown that for any $ v(0)\in \mathbb{V}_{0} $, the trajectory $ v(t) $ of the exosystem (\ref{exosystem}) remains in some compact set $ \mathbb{V} $ for all $ t\geq 0 $.
\end{remark}

%%====================================================================================================================================================================================
\section{Main Results}\label{section main results}

In this section, based on a two-layer controller structure, a distributed output feedback controller is developed to solve Problem \ref{problem}. In the upper layer, the concerned DOOC problem is first converted to a reference-tracking problem by designing a distributed optimal coordinator for each agent, which cooperates with others to generate a local reference signal that eventually converges to the optimal solution. Then in the lower layer, by virtue of an internal model for handling the external disturbance and a high-gain observer for estimating the agent states, a decentralized output feedback stabilizer is proposed to address the augmented system composed of the resulting error system and the internal model.

\subsection{Distributed Optimal Coordinator Design}
In this subsection, the distributed optimal coordinator is designed to generate the optimal solution for the concerned multi-agent system over the unbalanced directed network. Then the proposed optimal coordinator is embedded into the feedback loop to convert the DOOC problem to a reference-tracking problem, which will be addressed in the next subsection. Specifically, inspired by \cite{Zhu2018continuous,zhang2021optimal}, the optimal coordinator for each agent $ i $ is designed as follows,
\begin{align} 
	\begin{split} \label{algorithm_r}
		&\dot{y}_{i}^{r} \!=-\frac{1}{\xi_{i}^{i}} \nabla c_{i}\left(y_{i}^{r}\right)-\alpha_{1} \sum_{i=1}^{N} a_{i j}\left(y_{i}^{r}-y_{j}^{r}\right)-\alpha_{2} \zeta_{i}, \\[-1mm]
		&\dot{\zeta}_{i} =\alpha_{1} \sum_{i=1}^{N} a_{i j}\left(y_{i}^{r}-y_{j}^{r}\right), \quad \zeta_{i}(0)=0,
	\end{split}  \quad ~~\!\Bigggl\}\\
	\begin{split} \label{algorithm_r_b}
		&\dot{\xi}_{i}  = -\sum_{j=1}^{N}a_{ij}(\xi_{i}-\xi_{j}),\quad i=1,2,\ldots,N,
	\end{split}  \qquad \qquad \quad \!\Big\}
\end{align}
where $ \alpha_{1} $ and $ \alpha_{2} $ are two positive constants, $ y_{i}^{r}\in \mathbb{R} $ represents the generated reference signal for agent $ i $, $ \zeta_{i}\in \mathbb{R} $  and $ \xi_{i}\in\mathbb{R}^{N} $ are auxiliary variables, with $ \xi_{i}^{k} $ being the $ k $-th component of $ \xi_{i} $ and initial value $ \xi_{i}(0) $ satisfying
$ \xi_{i}^{i}(0) = 1 $, otherwise $ \xi_{i}^{k}(0) = 0 $ for all $ k \neq i $. With these choices, it is shown in \cite{zhang2021optimal} that $ \xi_{i}^{i}(t)>0 $ for all $t \geq 0$, which means that the algorithm (\ref{algorithm_r}) is well defined.

It should be emphasized that, under Assumption \ref{graph assumption}, it can be proved that $ \lim_{t \rightarrow \infty}\xi_{i}^{i}= r_{i} $, where $ r_{i} $ represents the $ i $-th component of the left eigenvector $ r $ corresponding to the zero eigenvalue of the Laplacian matrix \cite{zhang2021optimal}. Note that the optimal coordinator of each agent only requires information from its neighbors and itself, and thus is distributed.

\begin{remark}
	Based on the graph balancing technique developed in \cite{mei2015distributed,Bullo2019lectures}, a distributed optimization algorithm similar to the optimal coordinator (\ref{algorithm_r}), but with its gradient term being divided by the $ i $-th component of the left eigenvector $ r $ corresponding to the zero eigenvalue of the Laplacian matrix, is proposed in \cite{Li2017distributed}. However, it is worth noting that the left eigenvector is a kind of global information, and thus may not be known\textit{ a priori}. Instead, the gradient term $ \nabla c_{i} $ in (\ref{algorithm_r}) is divided by $ \xi_{i}^{i} $, which is generated by the autonomous system (\ref{algorithm_r_b}) and provides an alternative so that the explicit dependence on the left eigenvector can be removed.
\end{remark}

The following theorem shows that the distributed optimal coordinator (\ref{algorithm_r}) and (\ref{algorithm_r_b}) is capable of generating the optimal solution $ s^{\star} $. Due to space limitations, its proof is omitted here. One may refer to Theorem 1 in our preliminary work \cite{zhang2021optimal} for details.

\begin{theorem} \label{proposition1} 
	Consider the distributed optimal coordinator (\ref{algorithm_r}) and (\ref{algorithm_r_b}) under Assumptions \ref{assumption_cost functions} and \ref{graph assumption}. Given $ \zeta_{i}(0)=0 $ and any initial conditions $ y_{i}^{r}(0) $ for $ i=1,2,\ldots,N $, there exist sufficiently large positive constants  $ \alpha_{1}$ and $ \alpha_{2} $ such that the generated reference signals $ y_{i}^{r}, i=1,2,\ldots,N $ are bounded for all $ t\geq 0 $, and exponentially converge to the optimal solution $ s^{\star} $ that minimizes the global cost function.
\end{theorem}

In what follows, Problem \ref{problem} will be converted to a reference-tracking problem by taking $ y_{i}^{r} $ as a reference to be tracked by agent $ i $. To this end, define $x_{i}^{\star}=\operatorname{col}(y_{i}^{r}, \mathbf{0}_{n_{i}-1})$. Let $\bar{z}_{i}=z_{i}-z_{i}^{\star}\left(s^{\star},v, w\right)$ and $ \bar{x}_i=x_{i}-x_{i}^{\star} $. Then, the dynamics of $ \bar{z}_{i} $ and $ \bar{x}_{i} $ can be described as follows, 
\begin{equation} \label{error_system}
	\begin{aligned}
		\dot{\bar{z}}_{i}= & \bar{f}_{i 0}\left(\bar{z}_{i}, \bar{x}_{i 1}, y_{i}^{r}, v, w\right)+\hat{f}_{i 0}\left(y_{i}^{r}, s^{\star}, v, w\right), \\
		\dot{\bar{x}}_{i 1}= & \bar{x}_{i 2}-\dot{y}_{i}^{r}, \\
		\dot{\bar{x}}_{i k}= & \bar{x}_{i(k+1)}, \quad k=2,3, \ldots n_{i}-1, \\
		\dot{\bar{x}}_{i n_{i}} = & \bar{f}_{i 1}\left(\bar{z}_{i}, \bar{x}_{i}, y_{i}^{r}, v, w\right)+\bar{f}_{i 2}\left(y_{i}^{r}, s^{\star}, v, w\right)\\
		&+b_{i}(w)\left(u_{i}-u_{i}^{\star}\left(s^{\star}, v, w\right)\right),
	\end{aligned}
\end{equation}
where 
\begin{align*}
	\bar{f}_{i 0}\left(\bar{z}_{i}, \bar{x}_{i 1}, y_{i}^{r}, v, w\right)=&f_{i 0}\left(\bar{z}_{i}+z_{i}^{\star}, \bar{x}_{i 1}+y_{i}^{r}, v, w\right) \\
	&-f_{i 0}\left(z_{i}^{\star}, y_{i}^{r}, v, w\right),\\
	\hat{f}_{i 0}\left(y_{i}^{r}, s^{\star}, v, w\right)=&f_{i 0}\left(z_{i}^{\star}, y_{i}^{r}, v, w\right)-f_{i 0}\left(z_{i}^{\star}, s^{\star}, v, w\right),\\
	\bar{f}_{i 1}\left(\bar{z}_{i}, \bar{x}_{i}, y_{i}^{r}, v, w\right)= & f_{i 1}\left(\bar{z}_{i}+z_{i}^{\star}, \bar{x}_{i}+x_{i}^{\star}, v, w\right)\\
	&-f_{i 1}\left(z_{i}^{\star}, x_{i}^{\star}, v, w\right),\\
	\bar{f}_{i 2}\left(y_{i}^{r}, s^{\star}, v, w\right)=& f_{i 1}\left(z_{i}^{\star}, y_{i}^{r}, \mathbf{0}, v, w\right)- f_{i 1}\left(z_{i}^{\star}, s^{\star}, \mathbf{0}, v, w\right),\\
	u_{i}^{\star}\left(s^{\star}, v, w\right)=& -f_{i 1}\left(z_{i}^{\star}, s^{\star}, \mathbf{0}, v, w\right) / b_{i}(w).
\end{align*}
Recalling that $ b_{i}(w)>0 $ for all $ w\in\mathbb{W} $, $ u_{i}^{\star}\left(s^{\star}, v, w\right) $ is well-defined. Moreover, it follows from the smoothness of $ f_{i k}$, $ k=0,1 $ that $ \bar{f}_{i k}$, $k=0,1$ are sufficiently smooth functions satisfyinig $ \bar{f}_{i k}\left(\mathbf{0},\mathbf{0}, y_{i}^{r}, v,w\right)=0$ for all $ y_{i}^{r}\in\mathbb{R} $ and $ \operatorname{col}(v,w)\in \mathbb{V}\times \mathbb{W} $. Due to the presence of uncertain parameter $ w $, the feed-forward term $ u_{i}^{\star}\left(s^{\star}, v,w\right) $ in (\ref{error_system}) is unavailable. To tackle this challenge, we need an additional standard assumption \cite{huang2004nonlinear}.

\begin{assumption} \label{assumption_u}
	For $ i=1,2,\ldots,N $, $u_{i}^{\star}\left(s^{\star}, v,w\right)  $ is a polynomial in $v$ with coefficients depending on $ s^{\star} $ and $w$.
\end{assumption}

Now, we are ready to design an internal model to generate the feed-forward term $ u_{i}^{\star}\left(s^{\star}, v,w\right) $. Specifically, under Assumptions \ref{assumption_exosysem} and \ref{assumption_u}, there exist integers $s_{i}, i=1,2, \ldots, N$ such that for all $w \in \mathbb{W} $, one has $ \frac{\mathrm{d}^{s_{i}} u_{i}^{\star}\left(s^{\star}, v,w\right)}{\mathrm{d} t^{s_{i}}}= \ell_{i 1} u_{i}^{\star}\left(s^{\star}, v,w\right)+\ell_{i 2} \frac{\mathrm{d} u_{i}^{\star}\left(s^{\star}, v,w\right)}{\mathrm{d} t} +\cdots+\ell_{i s_{i}} \frac{\mathrm{d}^{\left(s_{i}-1\right)} u_{i}^{\star}\left(s^{\star}, v,w\right)}{\mathrm{d} t^{\left(s_{i}-1\right)}} $, 
where $\ell_{i 1},\ell_{i 2}, \ldots \ell_{i s_{i}}$, $i=1,2, \ldots, N$ are scalars such that the roots of the polynomials $P_{i}(\sigma)=\sigma^{s_{i}}-\ell_{i 1}-\ell_{i 2} \sigma-\cdots-\ell_{i s_{i}} \sigma^{s_{i}-1}$ are distinct with zero real parts. Define
\begin{align*}
	\Phi_{i}=\left[\begin{array}{c|c}
		\mathbf{0}_{\left(s_{i}-1\right) \times 1} & I_{s_{i}-1} \\
		\hline \ell_{i 1} & \ell_{i 2}, \ldots, \ell_{i s_{i}}
	\end{array}\right], 
	\quad \Gamma_{i}=\left[\begin{array}{c}
		1 \\
		\mathbf{0}_{\left(s_{i}-1\right) \times 1}
	\end{array}\right]^{T}.
\end{align*}
Let $\tau_{i}(s^{\star}, v,w)=\operatorname{col}\big(u_{i}^{\star}, \mathrm{d} u_{i}^{\star}/\mathrm{d} t, \ldots, \mathrm{d}^{\left(s_{i}-1\right)} u_{i}^{\star}/\mathrm{d} t^{\left(s_{i}-1\right)}\big)$. One thus has $ \dot{\tau}_{i}=\Phi_{i} \tau_{i}, u_{i}^{\star}=\Gamma_{i} \tau_{i} $. Let $\left(M_{i}, N_{i}\right), i=1,2, \ldots, N$ be any controllable pairs, where $M_{i} \in \mathbb{R}^{s_{i} \times s_{i}}$ is a Hurwitz matrix, and $N_{i} \in \mathbb{R}^{s_{i} \times 1}$ is a column vector. Then the following internal model is proposed,
\begin{equation} \label{internal_model}
	\dot{\eta}_{i}=M_{i} \eta_{i}+N_{i} u_{i}, \quad i=1,2, \ldots, N.
\end{equation}

Since the spectra of $\Phi_{i}$ and $M_{i}$ are disjoint, there exists a nonsingular matrix $T_{i}$ satisfying the Sylvester equation $ T_{i} \Phi_{i}-M_{i} T_{i}=N_{i} \Gamma_{i} $.
Let $\bar{\eta}_{i}=\eta_{i}-T_{i} \tau_{i}\left(s^{\star}, v,w\right)$ and $\bar{u}_{i}=u_{i}-\Gamma_{i} T_{i}^{-1} \eta_{i}$. One then has $ \dot{\bar{\eta}}_{i}=\left(M_{i}+N_{i} \Gamma_{i} T_{i}^{-1}\right) \bar{\eta}_{i}+N_{i} \bar{u}_{i} $.
Thus, we can obtain the following augmented error system,
\begin{equation} \label{augmented_system}
	\begin{aligned}
		\dot{\bar{z}}_{i}= &\bar{f}_{i 0}\left(\bar{z}_{i}, \bar{y}_{i}, y_{i}^{r}, v, w\right)+\hat{f}_{i 0}\left(y_{i}^{r}, s^{\star}, v, w\right), \\
		\dot{\bar{x}}_{i}= & A_{i}\bar{x}_{i}+B_{i}\Big[\bar{f}_{i 1}\left(\bar{z}_{i}, \bar{x}_{i}, y_{i}^{r}, v, w\right)+\bar{f}_{i 2}\left(y_{i}^{r}, s^{\star}, v, w\right)\\
		&+b_{i}(w) \Gamma_{i} T_{i}^{-1} \bar{\eta}_{i}+b_{i}(w) \bar{u}_{i}\Big]-E_{i}\dot{y}_{i}^{r},\\
		\dot{\bar{\eta}}_{i}=&\left(M_{i}+N_{i} \Gamma_{i} T_{i}^{-1}\right) \bar{\eta}_{i}+N_{i} \bar{u}_{i},\\
		\bar{y}_{i}=&\bar{x}_{i 1},
	\end{aligned}
\end{equation}
where
\begin{equation*}
	A_{i}=\Big[\begin{array}{c|c}
		\mathbf{0}_{n_{i}-1} & I_{n_{i}-1} \\
		\hline 0 & \mathbf{0}_{n_{i}-1}^{\operatorname{T}}
	\end{array}\Big], 
	B_{i}=\Big[\begin{array}{c}
		\mathbf{0}_{n_{i}-1} \\
		1
	\end{array}\Big], 
	E_{i}=\Big[\begin{array}{c}
		1 \\
		\mathbf{0}_{n_{i}-1}
	\end{array}\Big].
\end{equation*}

\vspace{2mm}
In what follows, we will develop a distributed dynamic output feedback controller of the following form to stabilize the augmented error system (\ref{augmented_system}) semi-globally,
\begin{equation} \label{controller form_error system}
	\begin{aligned}
		\bar{u}_{i} =  \beta_{\delta}\big(\bar{\kappa}_{i1}(\bar{x}_{i},\bar{\upsilon}_{i})\big),\quad
		\dot{\bar{\upsilon}}_{i} =  \bar{\kappa}_{i2}\big(y_{i}, \bar{\upsilon}_{j}, j\in \bar{\mathcal{N}}_{i}\big),
	\end{aligned}
\end{equation}
with $ \beta_{\delta}(\cdot) $ being a saturation function defined as follows,
\begin{equation*}
	\beta_{\delta}(r)= \begin{cases}r & \text { if }|r|<\delta, \\ \operatorname{sgn}(r) \delta & \text { if }|r| \geq \delta,\end{cases}
\end{equation*}
where $  \delta > 0 $ is a constant to be designed, $ \bar{\kappa}_{i1} $ and $ \bar{\kappa}_{i2} $ are sufficiently smooth functions vanishing at the origin, $ \bar{\mathcal{N}}_{i} = \mathcal{N}_{i} \cup \{i\} $ is defined to be the same as that in (\ref{controller form_origin system}), and $ \bar{\upsilon}_{i}\in \mathbb{R}^{n_{\bar{\upsilon}_{i}}} $ is the state of the dynamic controller with its dimention $ n_{\bar{\upsilon}_{i}} $ to be specificed later. Let $ \bar{x}_{c}=\operatorname{col}(\bar{z}_{1},\bar{x}_{1},\bar{\eta}_{1},\bar{\upsilon}_{1},y_{1}^{r},\zeta_{1},\ldots,\bar{z}_{N},\bar{x}_{N},\bar{\eta}_{N},\bar{\upsilon}_{N},y_{N}^{r},\zeta_{N}) $ and $ n_{c}=\sum_{i=1}^{N}(n_{z_{i}}+n_{i}+s_{i}+n_{\bar{\upsilon}_{i}}+2) $. The reference-tracking problem is defined as follows.

\vspace{2mm}
\begin{problem} \label{problem_2}
	 Consider the augmented error system (\ref{augmented_system}) and the optimal coordinator (\ref{algorithm_r}). Under Assumptions 1-4, given any constant $ \bar{R}>0 $ and any nonempty compact set $ \mathbb{V}\times\mathbb{W}\subseteq \mathbb{R}^{n_{v}+n_{w}} $ containing the origin, design a distributed dynamic output feedback controller of the form (\ref{controller form_error system}) such that, for any $\bar{x}_{c}(0)\in\bar{Q}_{\bar{R}}^{n_{c}} $ and $ \operatorname{col}(v,w)\in \mathbb{V}\times\mathbb{W} $, the trajectories of the closed-loop system composed of (\ref{algorithm_r}), (\ref{augmented_system}) and (\ref{controller form_error system}) starting from $ \bar{x}_{c}(0) $ are bounded for all $ t\geq 0 $, and the output errors $ \bar{y}_{i}, i=1,2,\ldots,N $ tend to zero as time goes to infinity.
\end{problem}

\vspace{2mm}
The following lemma shows that Problem \ref{problem} is solved as long as Problem \ref{problem_2} is solved.

\begin{lemma} \label{lemma_stabilization}
	Under Assumptions \ref{assumption_cost functions}-\ref{assumption_exosysem}, if the reference-tracking Problem \ref{problem_2} is solved by a dynamic output feedback controller of the form (\ref{controller form_error system}), then the Problem \ref{problem} can be solved by a distributed dynamic controller composed of (\ref{algorithm_r}), (\ref{internal_model}) and (\ref{controller form_error system}).
\end{lemma}

\begin{proof}
	As stated in Remark \ref{remark_d}, given any compact set $ \mathbb{V}_{0} \subseteq \mathbb{R}^{n_{v}} $, there exists a compact set $ \mathbb{V} \subseteq \mathbb{R}^{n_{v}} $ such that $ v(t)\in \mathbb{V} $ for all $ t\geq 0 $. By comparing the definitions of $ x_{c} $ and $ \bar{x}_{c} $, given any $R>0$ and $x_{c}(0) \in \bar{Q}_{R}^{n_{c}}$, there exists $\bar{R}>0$ such that $\bar{x}_{c}(0) \in \bar{Q}_{\bar{R}}^{n_{c}}$. Note that Problem \ref{problem_2} is solved means that the trajectory $ \bar{x}_{c}(t) $ of the closed-loop system starting from $ \bar{x}_{c}(0)\in\bar{Q}_{\bar{R}}^{n_{c}} $ is bounded for all $ t\geq 0 $, and the output errors $ \bar{y}_{i}, i=1,2,\ldots,N $ converge to zero. Thus, on the one hand, by similar arguments as in \cite{su2014cooperative}, for any $x_{c}(0)\in\bar{Q}_{R}^{n_{c}} $ and $ \operatorname{col}(v,w)\in \mathbb{V}\times\mathbb{W} $, $ x_{c}(t) $ can be shown to be bounded for all $ t\geq 0 $. On the other hand, one obtains that $ \lim_{t\to\infty} y_{i} = y_{i}^{r} $. It follows from Theorem \ref{proposition1} that $ y_{i}^{r} $ is bounded for all $ t\geq 0 $ and $ \lim_{t\to\infty} y_{i}^{r} = s^{\star} $. Therefore, by the triangle inequality $ |y_{i}-s^{\star}|\leq |y_{i}-y_{i}^{r}| + |y_{i}^{r}-s^{\star}| $, one can conclude that $ \lim_{t\to\infty} y_{i} = s^{\star} $. In summary, Problem \ref{problem} is solved as long as Problem \ref{problem_2} is solved.
\end{proof}

\subsection{Reference-Tracking Controller Design}

In this subsection, we focus on developing an output feedback controller to solve the reference-tracking Problem \ref{problem_2} for the augmented error system (\ref{augmented_system}). 
At first, we transfer (\ref{augmented_system}) to a new system of relative degree one. To this end, define
\begin{align*}
	\hat{x}_{i k}= & \tfrac{\bar{x}_{i k}}{g^{k-1}}, \quad k=1, \ldots, n_{i}-1,\\
	\vartheta_{i}= & \bar{x}_{i n_{i}}+ g \gamma_{i (n_{i}-1)} \bar{x}_{i (n_{i}-1)}+\cdots +g^{n_{i}-1} \gamma_{i 1} \bar{x}_{i 1},\\
	\tilde{\eta}_{i}= & \bar{\eta}_{i}-b_{i}^{-1}(w) N_{i} \vartheta_{i},
\end{align*}
where $g>0$ is a constant to be determined later, and the coefficients $\gamma_{i k}, k=1,2, \ldots, n_{i}-1 $ are chosen such that the polynomials $\lambda^{n_{i}-1}+\gamma_{i \left(n_{i}-1\right)} \lambda^{n_{i}-2}+\cdots+\gamma_{i 2} \lambda+\gamma_{i 1}, i=1,2,\ldots, N$ are all Hurwitz. Then, one can obtain that
\begin{align*}
	\dot{\hat{x}}_{i 1}=&g \hat{x}_{i 2}-\dot{y}_{i}^{r},\\
	\dot{\hat{x}}_{i k}=&g \hat{x}_{i(k+1)}, \quad k=2,3, \ldots, n_{i}-2,\\
	\dot{\hat{x}}_{i\left(n_{i}-1\right)} 
	=&\tfrac{\vartheta_{i}}{g^{n_{i}-2}}-g \gamma_{i\left(n_{i}-1\right)} \hat{x}_{i\left(n_{i}-1\right)}-\cdots-g \gamma_{i 1} \hat{x}_{i 1},\\
	\dot{\vartheta}_{i}=&\hat{f}_{i 1}\left(\bar{z}_{i}, \hat{x}_{i a}, \vartheta_{i}, y_{i}^{r}, g, v, w\right)+b_{i}(w) \Gamma_{i} T_{i}^{-1} \tilde{\eta}_{i}\\
	&+\Gamma_{i} T_{i}^{-1} N_{i} \vartheta_{i}+b_{i}(w) \bar{u}_{i}+\varepsilon_{i}\left(y_{i}^{r}, \dot{y}_{i}^{r}, s^{\star}, v, w\right),
\end{align*}
where
\begin{small}
	\begin{align*}
		&\hat{f}_{i 1}\left(\bar{z}_{i}, \hat{x}_{i a}, \vartheta_{i}, y_{i}^{r}, g, v, w\right)=\bar{f}_{i 1}\Big(\bar{z}_{i}, \hat{x}_{i 1}, g \hat{x}_{i 2}, \ldots, g^{n_{i}-2} \hat{x}_{i\left(n_{i}-1\right)},\\ &~~~~~~~\vartheta_{i}-g^{n_{i}-1}\left(\gamma_{i\left(n_{i}-1\right)} \hat{x}_{i\left(n_{i}-1\right)}+\cdots+\gamma_{i 1} \hat{x}_{i 1}\right), y_{i}^{r}, v, w\Big)\\
		&~~~~~~~+g\gamma_{i\left(n_{i}-1\right)}\Big(\vartheta_{i}-g^{n_{i}-1}\left(\gamma_{i\left(n_{i}-1\right)} \hat{x}_{i\left(n_{i}-1\right)}+\cdots +\gamma_{i 1} \hat{x}_{i 1}\right)\Big)\\
		&~~~~~~~+g^{n_{i}}\left(\gamma_{i\left(n_{i}-2\right)} \hat{x}_{i\left(n_{i}-1\right)}+\cdots+\gamma_{i 2} \hat{x}_{i 3}+\gamma_{i 1} \hat{x}_{i 2}\right),\\[1.5mm]
		&\varepsilon_{i}\left(y_{i}^{r}, \dot{y}_{i}^{r}, s^{\star}, v, w\right)=\bar{f}_{i 2}\left(y_{i}^{r}, s^{\star}, v, w\right)-g^{\left(n_{i}-1\right)} \gamma_{i 1} \dot{y}_{i}^{r}.
	\end{align*}
\end{small}

Let $\hat{x}_{i a}=\operatorname{col}\left(\hat{x}_{i 1}, \hat{x}_{i 2}, \ldots \hat{x}_{i\left(n_{i}-1\right)}\right)$. Then the augmented error system (\ref{augmented_system}) can be rewritten as follows,
\begin{equation} \label{augmented_system_1}
	\begin{aligned}
		\dot{\bar{z}}_{i}= &\tilde{f}_{i 0}\left(\bar{z}_{i}, \hat{x}_{i a}, y_{i}^{r}, v, w\right)+\hat{f}_{i 0}\left(y_{i}^{r}, s^{\star}, v, w\right), \\
		\dot{\hat{x}}_{i a}= & g A_{ci} \hat{x}_{i a}+B_{ci}(g) \vartheta_{i}-E_{i}\dot{y}_{i}^{r},\\
		\dot{\tilde{\eta}}_{i}=&M_{i} \tilde{\eta}_{i}+\tilde{f}_{i 1}\left(\bar{z}_{i}, \hat{x}_{i a}, \vartheta_{i}, y_{i}^{r}, g, v, w\right)-b_{i}^{-1}(w) N_{i} \varepsilon_{i}, \\
		\dot{\vartheta}_{i}=&\tilde{f}_{i 2}\left(\bar{z}_{i}, \hat{x}_{i a}, \tilde{\eta}_{i}, \vartheta_{i}, y_{i}^{r}, g, v, w\right)+b_{i}(w) \bar{u}_{i}+\varepsilon_{i},
	\end{aligned}
\end{equation}
where
\begin{align*}
	&A_{ci}=\left[\begin{array}{c|c}
		\mathbf{0}_{n_{i}-2} & I_{n_{i}-2} \\
		\hline-\gamma_{i 1} & -\gamma_{i 2}, \ldots,-\gamma_{i \left(n_{i}-1\right)}
	\end{array}\right],\\ 
	&B_{ci}(g)=\left[\begin{array}{c}
		\mathbf{0}_{n_{i}-2} \\
		1 / g^{n_{i}-2}
	\end{array}\right], \quad E_{i}=\Big[\begin{array}{c}
	1 \\
	\mathbf{0}_{n_{i}-1}
	\end{array}\Big],\\
	&\tilde{f}_{i 0}\left(\bar{z}_{i}, \hat{x}_{i a}, y_{i}^{r}, v, w\right)= \bar{f}_{i 0}\left(\bar{z}_{i}, \bar{x}_{i 1}, y_{i}^{r}, v, w\right),\\
	&\tilde{f}_{i 1}\left(\bar{z}_{i}, \hat{x}_{i a}, \vartheta_{i}, y_{i}^{r}, g, v, w\right)= b_{i}^{-1}(w) M_{i} N_{i} \vartheta_{i}\\
	&~~~~~~~~~~~~~~~~~~~~~~~-b_{i}^{-1}(w) N_{i} \hat{f}_{i 1}\big(\bar{z}_{i}, \hat{x}_{i a}, \vartheta_{i}, y_{i}^{r}, g, v, w\big),\\
	&\tilde{f}_{i 2}\left(\bar{z}_{i}, \hat{x}_{i a}, \tilde{\eta}_{i}, \vartheta_{i}, y_{i}^{r}, g, v, w\right)=\hat{f}_{i 1}\left(\bar{z}_{i}, \hat{x}_{i a}, \vartheta_{i}, y_{i}^{r}, g, v, w\right)\\
	&~~~~~~~~~~~~~~~~~~~~~~~+b_{i}(w) \Gamma_{i} T_{i}^{-1} \tilde{\eta}_{i}+\Gamma_{i} T_{i}^{-1} N_{i} \vartheta_{i}.
\end{align*}

Now we are ready to design the decentralized output feedback controller $ \bar{u}_{i} $ for the augmented system (\ref{augmented_system_1}) under the following assumption.

\begin{assumption} \label{assumption_zero dynamics_1}
	For each $ i=1,2,\ldots,N $, there exists a continuously differentiable function $V_{i 0}:\mathbb{R}^{n_{z_{i}}} \rightarrow \mathbb{R}$ such that, for all $ y_{i}^{r}\in\mathbb{R} $ and $\operatorname{col}(v, w) \in \mathbb{V}\times \mathbb{W}$, along the trajectories of the zero dynamics $\dot{\bar{z}}_{i}=\bar{f}_{i 0}(\bar{z}_{i},0,y_{i}^{r},v,w) $, the following inequality is satisfied,
	\begin{align}\label{locally exponential stability}
		\frac{\partial V_{i 0}}{\partial \bar{z}_{i}} \bar{f}_{i 0}(\bar{z}_{i},0,y_{i}^{r},v,w) \leq -a_{i 0}\left\|\bar{z}_{i}\right\|^{2},
	\end{align}
	where $ a_{i 0}>0 $ is a known constant.	
\end{assumption}

\begin{remark}
	According to Theorem 4.10 in \cite{khalil2002nonlinear}, the inequality (\ref{locally exponential stability}) implies that the zero dynamics of agents are globally asymptotically stable as well as locally exponentially stable, which is less stringent than the assumption that the inverse dynamics of agents are ISS in \cite{tang2020optimal,li2020distributed}. In particular, the locally exponential stability is needed to drive all trajectories to the origin instead of its neighborhood \cite{isidori1999nonlinear}. By choosing sufficiently large control gains, it is shown by semi-global stability analysis that $ \operatorname{col}(\bar{z}_{i}, \bar{x}_{i 1}) $ can still be guaranteed to remain in a compact set under such a relaxed assumption.
\end{remark}

\subsubsection{High-Gain Observer Design}
Due to the uncertainties in agent dynamics (\ref{dynamics}), it is more challenging to estimate the agent states since no perfect knowledge of the nonlinear functions $ f_{i 1}, i=1,2,\ldots,N $ can be utilized for designing an observer. Fortunately, the high gain observer is robust to a certain level of uncertainties and can still work in this case. Specifically, a distributed high-gain observer is proposed for each agent as follows,
\begin{equation} \label{high-gain_observer}
	\left(\!\begin{array}{c}
		\dot{\tilde{x}}_{i 1} \\
		\dot{\tilde{x}}_{i 2} \\
		\vdots \\
		\dot{\tilde{x}}_{i\left(n_{i}-1\right)} \\
		\dot{\tilde{x}}_{i n_{i}}
	\end{array}\!\right)=\left(\begin{array}{c}
		\tilde{x}_{i 2} \\
		\tilde{x}_{i 3} \\
		\vdots \\
		\tilde{x}_{i n_{i}} \\
		0
	\end{array}\right)+\left(\!\begin{array}{c}
		h c_{i n_{i}} \\
		h^{2} c_{i\left(n_{i}-1\right)} \\
		\vdots \\
		h^{n_{i}-1} c_{i 2} \\
		h^{n_{i}} c_{i 1}
	\end{array}\!\right)\left(y_{i}-\tilde{x}_{i 1}\right),
\end{equation}
where $ \tilde{x}_{i k}, k=1,2,\ldots,n_{i} $ are the estimations of the agent states, $ c_{i k}, k=1,2,\ldots,n_{i} $ are cofficients such that the polynomials $ \lambda^{n_{i}}+c_{i n_{i}} \lambda^{n_{i}-1}+\cdots+c_{i 2} \lambda+c_{i 1}, i=1,2,\ldots,N $ are all Hurwitz, and $ h>0 $ is a sufficiently large constant to be determined later.

\subsubsection{Decentralized Output Feedback Stabilizer Design}

With the proposed distributed high-gain observer (\ref{high-gain_observer}), the decentralized output feedback stabilizer for the augmented system (\ref{augmented_system_1}) is developed as follows,
\begin{equation} \label{decentralized output feedback stabilizer}
	\bar{u}_{i}=-\beta_{\delta}\big(K \tilde{\vartheta}_{i}\big),
\end{equation}
where $ K $ and $ \delta$ are positive constants to be designed later, $\tilde{\vartheta}_{i}=\tilde{x}_{i n_{i}}+g \gamma_{i\left(n_{i}-1\right)} \tilde{x}_{i\left(n_{i}-1\right)}+\cdots+g^{n_{i}-2} \gamma_{i 2} \tilde{x}_{i 2}+g^{n_{i}-1} \gamma_{i 1}\left(\tilde{x}_{i 1}-y_{i}^{r}\right)$, with $ \tilde{x}_{i k},k=1,2,\ldots,n_{i} $ being generated by the high-gain observer (\ref{high-gain_observer}) and $ y_{i}^{r} $ being generated by the optimal coordinator (\ref{algorithm_r}).

\begin{remark}
	 The high-gain observer is exploited owing to its robustness in estimating the agent states in the presence of uncertainties. However, it may also exhibit an impulsive-like behavior known as the peaking phenomenon, which may even interact with nonlinearities leading to a finite escape time. Hence, the saturation function $ \beta_{\delta}(\cdot) $ is utilized in (\ref{decentralized output feedback stabilizer}) to aviod the peaking phenomenon and to ensure the control input in a compact set of interest.
\end{remark}

The solvability of Problem \ref{problem_2} is summarized as follows.

\begin{theorem} \label{theorem_2}
	Under Assumption \ref{assumption_zero dynamics_1}, given any arbitrarily large constant $\bar{R}>0$, there exist sufficiently large positive constants $g$, $ K $, $ \delta $ and $ h $ depending on $\bar{R}$ such that the equilibrium point of the closed-loop system consisting of (\ref{algorithm_r}), (\ref{augmented_system_1}), (\ref{high-gain_observer}) and (\ref{decentralized output feedback stabilizer}) is uniformly locally asymptotically stable with its region of attraction containing $\bar{Q}_{\bar{R}}^{n_{c}}$, where $n_{c}=\sum_{i=1}^{N}\left(n_{z_{i}}+2 n_{i}+s_{i}+2\right) $. In other words, the semi-global reference-tracking Problem \ref{problem_2} is solvable.
\end{theorem}

Before giving the proof of Theorem \ref{theorem_2}, some coordinate transformations are introduced to redescribe the resulting closed-loop system. Let $ \tilde{y}_{i}^{r}= y_{i}^{r}-\bar{y}_{i}^{r} $ and $ \tilde{\zeta}_{i}= \zeta_{i}-\bar{\zeta}_{i} $, where $ \operatorname{col}(y_{i}^{r}, \zeta_{i}) $ is generated by (\ref{algorithm_r}) with $ \operatorname{col}(\bar{y}_{i}^{r}, \bar{\zeta}_{i}) $ being its corresponding equilibrium point. Define $X_{i}=\operatorname{col}\left(\bar{z}_{i}, \hat{x}_{i a}, \tilde{\eta}_{i}, \vartheta_{i}, \tilde{y}_{i}^{r}, \tilde{\zeta}_{i}\right)$. Let $e_{i k}=h^{n_{i}-k}\left(x_{i k}-\tilde{x}_{i k}\right), k=1,2, \ldots, n_{i}$ and $e_{i}=\operatorname{col}\left(e_{i 1}, e_{i 2}, \ldots, e_{i n_{i}}\right)$. It follows that $\tilde{x}_{i}=x_{i}-H_{i}^{-1} e_{i}$, where $H_{i}=\operatorname{diag}\left(h^{n_{i}-1}, h^{n_{i}-2}, \ldots, h, 1\right)$. One then has $ \tilde{\vartheta}_{i}= \vartheta_{i}-E_{i 2}(h) e_{i} = E_{i 1}X_{i} - E_{i 2}(h) e_{i} $, where
\begin{align*}
	E_{i 1}=& \left[\mathbf{0}_{1 \times (n_{z_{i}}+s_{i}+n_{i}-1)},1,\mathbf{0}_{1 \times 2}\right],\\
	E_{i 2}(h)=&\left[\frac{g^{n_{i}-1} \gamma_{i 1}}{h^{n_{i}-1}}, \frac{g^{n_{i}-2} \gamma_{i 2}}{h^{n_{i}-2}}, \ldots, \frac{g \gamma_{i\left(n_{i}-1\right)}}{h}, 1\right].
\end{align*}
Thus, it can be obtained from (\ref{high-gain_observer}) and (\ref{decentralized output feedback stabilizer}) that
\begin{align}
	\bar{u}_{i}=&-\beta_{\delta}\big( KE_{i 1}X_{i} - KE_{i 2}(h) e_{i}\big), \label{bar_u_2}\\
 	\dot{e}_{i}=& h E_{i 3} e_{i}+E_{i 4} \dot{X}_{i},	\label{e_{i}_1}
\end{align}
where 
\begin{align} \label{E_i3}
	E_{i 3}=\left[\begin{array}{ccccc}
		-c_{i n_{i}} & 1 & 0 & \cdots & 0 \\
		-c_{i\left(n_{i}-1\right)} & 0 & 1 & \cdots & 0 \\
		\vdots & \vdots & \vdots & \ddots & \vdots \\
		-c_{i 2} & 0 & 0 & \cdots & 1 \\
		-c_{i 1} & 0 & 0 & 0 & 0
	\end{array}\right],
\end{align}
and $ E_{i 4} $ is given in (\ref{E_{i 4}}).
Thus, the closed-loop system composed of (\ref{algorithm_r}), (\ref{augmented_system_1}), (\ref{high-gain_observer}) and (\ref{decentralized output feedback stabilizer}) can be described by the closed-loop system composed of (\ref{algorithm_r}), (\ref{augmented_system_1}), (\ref{bar_u_2}) and (\ref{e_{i}_1}). Let $e=\operatorname{col}\left(e_{1},e_{2}, \ldots, e_{N}\right)$ and $X=\operatorname{col}\left(X_{1}, X_{2}, \ldots, X_{N}\right)$. Then, we can obtain the following two crucial lemmas, whose proofs are given in \hyperlink{appendices}{Appendices}.

\newcounter{mytempeqncnt}
\begin{figure*}[t]
	\vspace*{1pt}
	\hrulefill
	\hrule
	\normalsize
	\setcounter{mytempeqncnt}{\value{equation}}
	\begin{align} 
		E_{i 4}=&\left[\begin{array}{ccccc}
			\mathbf{0}_{\left(n_{i}-1\right) \times n_{z_{i}}} & \mathbf{0}_{\left(n_{i}-1\right) \times\left(n_{i}-1\right)} & \mathbf{0}_{\left(n_{i}-1\right) \times s_{i}} & \mathbf{0}_{\left(n_{i}-1\right) \times 1} & \mathbf{0}_{\left(n_{i}-1\right) \times 2} \\
			\mathbf{0}_{1 \times n_{z_{i}}} & -g^{n_{i}-1}\left[\gamma_{i 1}, \gamma_{i 2}, \ldots, \gamma_{i\left(n_{i}-1\right)}\right] & \mathbf{0}_{1 \times s_{i}} & 1 & \mathbf{0}_{1 \times 2} 
		\end{array}\right] \label{E_{i 4}}
	\end{align}
	\hrule
	\hrulefill
\end{figure*}

\begin{lemma} \label{Lemma_X}
Under Assumption \ref{assumption_zero dynamics_1}, given any arbitrarily large $\bar{R}>0$, there exist positive constants $g$, $ K $ and $c$ that depend on $\bar{R}$, and a continuously differentiable positive definite function $W_{X}(\cdot)$, such that the following equation is satisfied,
\begin{equation} \label{varOmega_c_X}
	\bar{Q}_{\bar{R}}^{\bar{n}_{c}} \subseteq \bar{\varOmega}_{c}\left(W_{X}(X)\right),
\end{equation}
where $ \bar{n}_{c}= \sum_{i=1}^{N}\left(n_{z_{i}}+n_{i}+s_{i}+2\right)$. Moreover, given any $ \epsilon>0 $, for all $X \in \bar{\varOmega}_{c+\epsilon}\big(W_{X}(X)\big)$ and all  $\operatorname{col}(v, w) \in \mathbb{V} \times \mathbb{W}$, the derivative of $W_{X}(\cdot)$ along the trajectories of the closed-loop system composed of (\ref{algorithm_r}), (\ref{augmented_system_1}) and the following decentralized state feedback controller,
\begin{equation} \label{state feedback}
	\bar{u}_{i}=-K \vartheta_{i},
\end{equation}
satisfies 
\begin{equation} \label{dot_W_X}
	\dot{W}_{X}(X)\big|_{(\ref{algorithm_r})+(\ref{augmented_system_1})+(\ref{state feedback})} \leq -a\|X\|^{2},
\end{equation}
where $ a>0 $ is a constant to be determined later.
\end{lemma}

\begin{lemma} \label{Lemma_e}
There exists a continuously differentiable positive definite function $W_{e}(e)$ such that, for all $\operatorname{col}(v, w) \in$ $\mathbb{V} \times \mathbb{W}$, the derivative of $W_{e}(e)$ along the trajectories of (\ref{e_{i}_1}) satisfies 
\begin{equation} \label{dot_W_e}
	\dot{W}_{e}(e)\big|_{(\ref{e_{i}_1})} \leq-\frac{h}{2}\|e\|^{2}+\frac{\varrho}{h}\|\dot{X}\|^{2},
\end{equation}
where $\varrho>0$ is a constant independent of $h$.	
\end{lemma}

Based on Lemmas \ref{Lemma_X} and \ref{Lemma_e}, the proof of Theorem \ref{theorem_2} is presented as follows.

\begin{proof}[Proof of Theorem \ref{theorem_2}]
	For any arbitrarily large $\bar{R}>0$, $g$ and $K$ are chosen to be the same as those in Lemma \ref{Lemma_X}. For any arbitrarily small $\epsilon>0$, the bound $ \delta $ of the saturation function is chosen as 
	\begin{equation} \label{saturation_S}
		\delta=K\bigg(\!\max _{X \in \bar{\varOmega}_{c+\epsilon}\left(W_{X}(X)\right) \atop i=1, \ldots, N}\!\!\left\{{E_{i 1}} X_{i}\right\}+\frac{\epsilon}{2} \max _{i=1, \ldots, N}\left\{\left\|E_{i 2}(1)\right\|\right\}\!\bigg) .
	\end{equation}
	Now, we are ready to prove that the trajectory of $\operatorname{col}(X(t), e(t))$ starting from any $ \operatorname{col}(X(0), e(0))\in\bar{Q}_{\bar{R}}^{n_{c}}$ will asymptotically converge to the origin by choosing sufficiently large $h$. To this end, consider the positive definite Lyapunov function $W(X, e)=W_{X}(X)+W_{e}(e)$, where $W_{X}(\cdot)$ and $W_{e}(\cdot)$ are defined to be the same as those in Lemmas \ref{Lemma_X} and \ref{Lemma_e}, respectively. The rest of the proof can be accomplished in the following two steps.
	
	\vspace{1mm}
	\textbf{Step 1.} Show that the trajectory of $\operatorname{col}(X(t), e(t))$ enters the compact set $\bar{\varOmega}_{c+\epsilon}(W(X, e))$ in a finite time $T$. By substituting the saturated output feedback controller (\ref{bar_u_2}) into (\ref{augmented_system_1}), one obtains that
	\begin{equation} \label{vartheta_i}
		\dot{\vartheta}_{i}=\tilde{f}_{i 2} -b_{i}(w)K\vartheta_{i} +\varepsilon_{i} + b_{i}(w)(\bar{u}_{i}+K\vartheta_{i}).
	\end{equation} 
	Note that $ 4 \varsigma_{i 6} \varsigma_{i 7} \varsigma_{i 8}(g)\geq 1 $, where $ \varsigma_{i 6} $, $ \varsigma_{i 7} $ and $ \varsigma_{i 8}(g) $ are chosen as the same as those in the proof of Lemma \ref{Lemma_X}. Then the derivative of $ V_{i 3}(\vartheta_{i})/(4 \varsigma_{i 6} \varsigma_{i 7} \varsigma_{i 8}(g)) $ along the trajectories of (\ref{vartheta_i}) can be described as follows,
	\begin{equation*}
		\tfrac{\dot{V}_{i 3}(\vartheta_{i})}{4 \varsigma_{i 6} \varsigma_{i 7} \varsigma_{i 8}(g)} \!\leq \tfrac{\vartheta_{i}(\tilde{f}_{i 2} -b_{i}(w)K\vartheta_{i} +\varepsilon_{i})}{4 \varsigma_{i 6} \varsigma_{i 7} \varsigma_{i 8}(g)}+\big|b_{i}(w)\vartheta_{i}(\bar{u}_{i}+K\vartheta_{i})\big|.
	\end{equation*}

	Reconsider $ W_{X}(X) $ in Lemma \ref{Lemma_X}. Its derivative along the trajectories of the closed-loop system composed of (\ref{algorithm_r}), (\ref{augmented_system_1}) and (\ref{bar_u_2}) can be described as follows,
	
	\begin{align} \label{dot_W_X_1}
		\dot{W}_{X}(X)|_{(\ref{algorithm_r})+(\ref{augmented_system_1})+(\ref{bar_u_2})} \leq& \dot{W}_{X}(X)|_{(\ref{algorithm_r})+(\ref{augmented_system_1})+(\ref{state feedback})}\notag\\
		& + \big|\sum_{i=1}^{N}b_{i}(w)\vartheta_{i}(\bar{u}_{i}+K\vartheta_{i})\big|.
	\end{align}

	On the one hand, it can be obtained from (\ref{bar_u_2}) that $\bar{u}_{i}$ is upper bounded by $\delta$, which is independent of $h$. Consequently, for all $X \in \bar{\varOmega}_{c+\epsilon}\left(W_{X}(X)\right)$ and all $\operatorname{col}(v, w) \in \mathbb{V} \times \mathbb{W}$, one has $\big|\sum_{i=1}^{N}b_{i}(w)\vartheta_{i}(\bar{u}_{i}+K\vartheta_{i})\big| \leq \delta_{1}$, where $\delta_{1}>0$ is a constant independent of $h$. By using (\ref{dot_W_X}), one further deduces that
	\begin{equation} \label{dot_W_X_u}
		\dot{W}_{X}(X)|_{(\ref{algorithm_r})+(\ref{augmented_system_1})+(\ref{bar_u_2})} \leq -a\|X\|^{2}+\delta_{1} \leq \delta_{1}.
	\end{equation}
	Meanwhile, it follows from (\ref{varOmega_c_X}) that $W_{X}(X(0)) \in \bar{\varOmega}_{c}\left(W_{X}(X)\right) \subseteq \bar{\varOmega}_{c+\epsilon}\left(W_{X}(X)\right)$. Thus, by using (\ref{dot_W_X_u}) and letting $T=\frac{\epsilon}{2 \delta_{1}}$, for all $t\in[0, T]$, we have
	\begin{equation} \label{W_X_t}
		W_{X}(X(t)) \leq W_{X}(X(0))+\delta_{1} t \leq c+\frac{\epsilon}{2}.
	\end{equation}
	In other words, it is shown in (\ref{W_X_t}) that $\|X(t)\|$ is upper bounded for all $t \in[0, T]$, with its bound independent of $h$. 
	
	On the other hand, it can be shown from the closed-loop system composed of (\ref{algorithm_r}), (\ref{augmented_system_1}) and (\ref{bar_u_2}) that $\varrho\|\dot{X}\|^{2} \leq \delta_{2}$, where $\delta_{2}>0$ is a constant independent of $h$. Thus, by referring to (\ref{dot_W_e}), one has $\dot{W}_{e}(e)\big|_{(\ref{e_{i}_1})} \leq -\frac{h}{2 \underline{\lambda}} W_{e}+\frac{\delta_{2}}{h}$, where $\underline{\lambda} \triangleq \min\limits_{i=1,2, \ldots, N}\{\lambda_{\min }\left(P_{i 3}\right), 1\} $. By the comparison lemma \cite{khalil2002nonlinear}, it follows that $W_{e}(e(t)) \leq \exp\{-\frac{h}{2 \underline{\lambda}}{t}\} W_{e}(e(0))+\frac{2 \underline{\lambda} \delta_{2}}{h^{2}}$, for all $t\in[0, T]$. Let $ h_{1}=\sqrt{\frac{8\underline{\lambda}\delta_{2}}{\epsilon}} $. Then, for all $t\in[0, T]$, choosing $ h\geq h_{1} $ leads to
	\begin{equation} \label{W_e_t}
		W_{e}(e(t)) \leq \exp\{-\tfrac{h}{2 \underline{\lambda}}{t}\} W_{e}(e(0))+ \frac{\epsilon}{4}.
	\end{equation}
	By recalling $e_{i j}=h^{n_{i}-j}\left(x_{i j}-\tilde{x}_{i j}\right)$, it can be proved that $W_{e}(e(0))$ is bounded by a polynomial of $h$. One then has $\lim\limits_{h \rightarrow+\infty} \exp\{-\frac{h}{2 \underline{\lambda}} T\} W_{e}(e(0))=0$. Thus, there exists a positive constant $h_{2}$ such that, for all $h>h_{2}$, $ \exp\{-\frac{h}{2 \underline{\lambda}} T\} W_{e}(e(0))<\frac{\epsilon}{4}$. By using (\ref{W_X_t}) and (\ref{W_e_t}), one can obtain that $W(X(T), e(T))<c+\epsilon$. In other words, the trajectory of
	$\operatorname{col}(X(t), e(t))$ enters the compact set $\bar{\varOmega}_{c+\epsilon}(W(X, e))$ within the time $T$.
	
	\textbf{Step 2.} Show that the trajectory of $\operatorname{col}(X(t), e(t))$ remains in the compact set $\bar{\varOmega}_{c+\epsilon}\left(W_{X}(X)\right) \times \bar{\varOmega}_{\epsilon/2}\left(W_{e}(e)\right)$, and asymptotically converges to the origin as time goes to infinity. It can be observered from (\ref{W_e_t}) that $ \exp\{-\tfrac{h}{2 \underline{\lambda}}{t}\} W_{e}(e(0)) $ is decreasing as time goes to infinity. Thus, by choosing any $h \geq \max \left\{h_{1}, h_{2}\right\}$, we have $W_{e}(e(t))<\epsilon/2$ for all $t \geq T$, that is, $e(t)$ remains in the compact set $\bar{\varOmega}_{\epsilon/2}\left(W_{e}(e)\right)$ for all $t \geq T$. Note that $\|E_{i 2}(h)\| \leq\|E_{i 2}(1)\|$ for all $h \geq 1$. It then follows from (\ref{saturation_S}) that $|\bar{u}_{i}| \leq \delta $, which implies that the saturation function will not trigger. As a consequence, the decentralized output feedback controller and the function in the resulting closed-loop system are smooth. Thus, it follows from Lemma 7.8 in \cite{huang2004nonlinear} that, for all $h \geq 1$, all $\operatorname{col}(v, w) \in \mathbb{V} \times \mathbb{W}$ and all $\operatorname{col}(X, e) \in \bar{\varOmega}_{c+\epsilon}\left(W_{X}(X)\right) \times \bar{\varOmega}_{\epsilon/2}\left(W_{e}(e)\right)$, the following inequalities are satisfied,
	\begin{align}
		\bigg|\sum_{i=1}^{N}b_{i}(w)\vartheta_{i}(\bar{u}_{i}+K\vartheta_{i})\bigg| \leq &\tfrac{\delta_{3}}{\varsigma}\|X\|^{2}+\varsigma \delta_{4}\|e\|^{2}, \label{inequality_1}\\
		\varrho\|\dot{X}\|^{2} \leq & \delta_{5}\|X\|^{2}+\delta_{6}\|e\|^{2}, \label{inequality_2}
	\end{align}
	where $ \varsigma $ and $ \delta_{l}, l=3,4,5,6 $ are some positive constants  independent of $ h $. Then by substituting (\ref{inequality_1}) and (\ref{inequality_2}) into (\ref{dot_W_X_1}) and (\ref{dot_W_e}) respectively, for all $\operatorname{col}(v, w) \in \mathbb{V} \times \mathbb{W}$ and all $\operatorname{col}(X, e) \in \bar{\varOmega}_{c+\epsilon}\left(W_{X}(X)\right) \times \bar{\varOmega}_{\epsilon/2}\left(W_{e}(e)\right)$, the derivative of $ W(X, e) $ along the trajectories of the closed-loop system consisting of (\ref{algorithm_r}), (\ref{augmented_system_1}), (\ref{bar_u_2}) and (\ref{e_{i}_1}) can be described as follows.
	\begin{align*}
		\dot{W}(X,e) \leq -\big(a-\tfrac{\delta_{3}}{\varsigma}-\tfrac{\delta_{5}}{h}\big)\|X\|^{2}- \big(\tfrac{h}{2}-\varsigma\delta_{4}-\tfrac{\delta_{6}}{h}\big)\|e\|^{2}.
	\end{align*}
	Choose $ \varsigma>\frac{\delta_{3}}{a} $ and $ h_{3}> \max\{\frac{\varsigma\delta_{5}}{a\varsigma-\delta_{3}}, 2(\varsigma\delta_{4}+\delta_{6}), 1\} $ successively. Then, for all $ h\geq \max\limits_{i=1,2,3}\{h_{i},1\} $, one can obtain that 
	\begin{align} \label{dot_Xe}
		\dot{W}(X,e)|_{(\ref{algorithm_r})+(\ref{augmented_system_1})+(\ref{bar_u_2})+(\ref{e_{i}_1})} \leq -\underline{a}\|(X,e)\|^{2},
	\end{align}
	where $ \underline{a}=\min\{a-\tfrac{\delta_{3}}{\varsigma}-\tfrac{\delta_{5}}{h},~\tfrac{h}{2}-\varsigma\delta_{4}-\tfrac{\delta_{6}}{h}\} $.
	
	Note that $\bar{\varOmega}_{c+\epsilon}(W(X, e)) \subseteq \bar{\varOmega}_{c+\epsilon}\left(W_{X}(X)\right) \times \bar{\varOmega}_{\epsilon/2}\left(W_{e}(e)\right)$. It can be concluded that, once the trajectory of $\operatorname{col}(X(t), e(t))$ enters $\bar{\varOmega}_{c+\epsilon}$ $(W(X, e))$, it remains in the compact set $\bar{\varOmega}_{c+\epsilon}\left(W_{X}(X)\right) \times \bar{\varOmega}_{\epsilon/2}\left(W_{e}(e)\right)$. In addition, it is shown in (\ref{dot_Xe}) that the trajectory of $\operatorname{col}(X(t), e(t))$ converges to the origin as time goes to infinity. Therefore, it is proved that the equilibrium point of the closed-loop system is semi-globally asymptotically stable, with its region of attraction containing $\bar{Q}_{\bar{R}}^{n_{c}}$.
\end{proof}

Based on Lemma \ref{lemma_stabilization} and Theorem \ref{theorem_2}, the main result of this work is summaried as follows.

\begin{theorem} \label{theorem 2}
	Under Assumptions \ref{assumption_cost functions}--\ref{assumption_zero dynamics_1}, given any $R>0$ and any compact set $\mathbb{V}_{0} \times \mathbb{W} \subseteq \mathbb{R}^{n_{v}+n_{w}}$, there exist sufficiently large $g$, $K$, $\delta$ and $h$ that depend on $R$ such that, for any $ x_{c}(0)\in \bar{Q}_{R}^{n_{c}} $ and any $ \operatorname{col}(v(0),w)\in\mathbb{V}_{0}\times\mathbb{W} $, the optimal output consensus Problem \ref{problem} is solved by the following distributed output feedback controller,
	\begin{align}
		u_{i} &=-\beta_{\delta}\big(K \tilde{\vartheta}_{i}\big) + \Gamma_{i} T_{i}^{-1} \eta_{i}, \label{controller_1_1}\\
		\dot{\eta}_{i} &=M_{i} \eta_{i}+N_{i} u_{i}, \quad i=1,2,\ldots,N, \label{controller_1_2}
	\end{align}
	where $\tilde{\vartheta}_{i}=\tilde{x}_{i n_{i}}+g \gamma_{i\left(n_{i}-1\right)} \tilde{x}_{i\left(n_{i}-1\right)}+\cdots+g^{n_{i}-2} \gamma_{i 2} \tilde{x}_{i 2}+g^{n_{i}-1} \gamma_{i 1}\left(\tilde{x}_{i 1}-y_{i}^{r}\right)$, with $ \tilde{x}_{i k},k=1,2,\ldots,n_{i} $ being generated by the distributed high-gain observer (\ref{high-gain_observer}) and $ y_{i}^{r} $ being generated by the distributed optimal coordinator (\ref{algorithm_r}).
\end{theorem}

\begin{remark}
	Due to the utilization of the saturated output feedback controller (\ref{controller_1_1}), the selection of the initial agent conditions is restricted in a positively invariant compact subset of the attraction region. Hence, only a semi-global stability of the resulting closed-loop system can be achieved by the proposed controller. It is worth noting that the controller (\ref{controller_1_1}) is designed by saturating the linear state feedback controller (\ref{state feedback}). Compared with the nonlinear controller developed in \cite{tang2020optimal}, a linear controller has many advantages in both theoretical design and practical implementation.
\end{remark}

%%==================================================================================================================================================================================
\section{An Example} \label{section simulation results}
In this section, an example is provided to illustrate the effectiveness of the proposed distributed output feedback controller. Consider a group of five agents with their dynamics described by the following uncertain nonlinear systems \cite{su2019semi},
\begin{align*}
	\dot{z}_{i 1} &=p_{i 1} z_{i 1}+x_{i 1} + A_{w 1}\sin(\theta t), \\
	\dot{x}_{i 1} &=x_{i 2} \\
	\dot{x}_{i 2} &=p_{i 2} z_{i 1} x_{i 1} x_{i 2}+p_{i 3} x_{i 2} + A_{w 2}\cos(\theta t) +b_{i}(w)u_{i}, \\
	y_{i} &=x_{i 1}, \quad i=1,2,
\end{align*}
and
\begin{align*}
	\dot{x}_{i 1} &=x_{i 2}, \\
	\dot{x}_{i 2} &=x_{i 3}, \\
	\dot{x}_{i 3} &=p_{i 1} x_{i 3}+p_{i 2} x_{i 2}+p_{i 3} x_{i 1}^{3}+ A_{w 3}\sin(\theta t) + b_{i}(w)u_{i}, \\
	y_{i} &=x_{i 1}, \quad i=3,4,5,
\end{align*}
where $ A_{w 1}\sin(\theta t) $, $ A_{w 2}\cos(\theta t) $ and $ A_{w 3}\sin(\theta t) $ are external disturbances, with $ A_{w k} = \mu_{i k}(w)A, ~k=1,2,3  $ being the uncertain amplitudes and $ \theta $ being the angular frequency. For $ i=1,\ldots,5 $, it is assumed that $ \mu_{i k}(w) = 1+0.1k, ~k=1,2,3 $. Note that $ \sin(\theta t)= v_{1}$ and $ \cos(\theta t)= v_{2} $, where $ v=\operatorname{col}(v_{1},v_{2}) $ is generated by the exosystem (\ref{exosystem}) with $S=[0, ~\theta;-\theta, ~0]$. Thus, Assumption \ref{assumption_exosysem} is satisfied. In addition, for $ i=1,\ldots,5 $, it is assumed that $ p_{i} = (p_{i 1}, p_{i 2}, p_{i 3}) $ satisfies $ p_{i} = \bar{p}_{i} + w_{i} $, where $ \bar{p}_{i}= (\bar{p}_{i 1}, \bar{p}_{i 2}, \bar{p}_{i 3}) $ denotes the nominal value of $ p_{i} $, $ w_{i} = (w_{i1}, w_{i2}, w_{i3}) \in \mathbb{W} $ denotes the uncertainty. Moreover, it can be verified that $ z_{i}^{\star}(s,v,w) $ in Assumption \ref{assumption_zero_dynamics} can be chosen as $ z_{i}^{\star}(s,v,w)=-\frac{p_{i1}A_{w1}}{p_{i1}^{2}+\theta^{2}}v_{1}-\frac{\theta A_{w1}}{p_{i1}^{2}+\theta^{2}}v_{2} +\frac{p_{i1}a_{i2}}{\theta}s, ~i=1,2 $. It is worth pointing out that the nonlinear functions in the agent dynamics are not globally Lipschitz.

\begin{figure}[!t] 
	\setlength{\abovecaptionskip}{-0.05cm}
	\centering 
	\begin{tikzpicture}[> = stealth, % arrow head style
		shorten > = 1pt, % don't touch arrow head to node
		auto,
		node distance = 3cm, % distance between nodereferens
		semithick % line style
		,scale=0.5,auto=left,every node/.style={circle,fill=gray!20,draw=black!80,text centered}]
		\centering
		\node (n1) at (0,0)		{1};
		\node (n2) at (2,0)  	{2};
		\node (n3) at (4,0) 	{3};
		\node (n4) at (6,0) 	{4};
		\node (n5) at (8,0) 	{5};
		
		\draw[->,black!80] (n3) to [out=-135,in=-45] (n1);
		\draw[->,black!80] (n1)-- (n2);
		\draw[->,black!80] (n2)-- (n3);
		\draw[->,black!80] (n3)-- (n4);
		\draw[->,black!80] (n4)-- (n5);
		\draw[->,black!80] (n2) to [out=-35,in=-145] (n5);
		\draw[->,black!80] (n5) to [out=150,in=30] (n1);
		
	\end{tikzpicture} 
	\caption{Weight-unbalanced directed network.} 
	\label{Fig_topology1}
\end{figure}
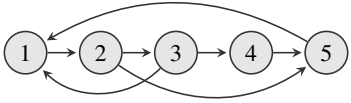

The unbalanced directed communication network among the agents is described in Fig. \ref{Fig_topology1}. It can be verified that the directed graph is strongly connected. For $ i=1,\ldots,5 $, suppose that each agent $ i $ is assigned with a local cost function $ c_{i}(s)=\frac{1}{4}(s-i+1)^{2} $. Then, all local cost functions are strongly convex, and the global minimizer is $ s^{\star}=2 $. Therefore, Assumptions \ref{assumption_cost functions} and \ref{graph assumption} are satisfied. It can be calculated that $u_{i}^{\star}\left(s^{\star}, v,w\right)=-b_{i}(w)^{-1} A_{w 2} v_{2}$ for $ i=1,2 $, and $u_{i}^{\star}\left(s^{\star}, v,w\right)=-b_{i}(w)^{-1} \big(p_{i 3}s^{\star 3} + A_{w 3} v_{1}\big)$ for $ i=3,4,5 $. Thus, Assumption \ref{assumption_u} is also satisfied. Specifically, one can deduce that  $\frac{\mathrm{d}^{2} u_{i}^{\star}\left(s^{\star}, v,w\right)}{\mathrm{d} t^{2}}=-\theta^{2} u_{i}^{\star}(s^{\star}, v,w)$ for $ i=1,2 $, and $\frac{\mathrm{d}^{3} u_{i}^{\star}\left(s^{\star}, v,w\right)}{\mathrm{d} t^{3}}=-\theta^{2} \frac{\mathrm{d} u_{i}^{\star}\left(s^{\star}, v,w\right)}{\mathrm{d} t}$ for $ i=3,4,5 $. Thus, by Theorem \ref{theorem 2}, the optimal output consensus problem for the concerned nonlinear multi-system is solvable. 

In the simulation, choose $ A=10 $ and $ \theta=0.8 $. Let $ \bar{p}_{i 1}=\bar{p}_{i 2} =\bar{p}_{i 3}= i $, and $ w_{i} = (w_{i1}, w_{i2}, w_{i3}) $ be randomly generated such that $ p_{i 1}, i=1,\ldots,5 $ are negative. Let $M_{i}=[0, 1;-2, -3], ~ N_{i}=[0; 1]$ for $i=1,2,$ and $M_{i}=[-3,-7,-5;~1, 0, 0;~0, 1, 0]$,  $N_{i}=[1; ~0; ~0]$ for $i=3,4,5$. The other initial conditions are randomly chosen. Then apply the distributed output feedback controller (\ref{controller_1_1}) with $ K=4\times10^{4} $ and $ h=100 $. The simulation result is presented in Fig. \ref{Fig_second_order_example1_1}. It can be observered from Fig. \ref{Fig_second_order_example1_1} that all the agent outputs $ y_{i}, i=1,\ldots,5 $ eventually converge to the optimal solution $ s^{\star}=2 $.

%%%%%%%%%%%%%%%%%%%%%%%%%%%%%%%%%%%%%%%%%%%%%%%%%%%%%%%%%%%%%%%%%%%%%%%%%%%%%%%%%%%%%%%%%%%%%%%%%%%%%%%%%%%%%%%%%%%%%%%%%%%%%%%%%%%%%%%%%%%%%%%%%%%%%%%%%%%%%%
\section{Conclusion} \label{section conclusion}
In this note, the distributed optimal output consensus problem of uncertain nonlinear multi-agent systems in the normal form over unbalanced directed networks is solved by a novel distributed output feedback controller. Based on a two-layer controller structure, the concerned problem is first converted to a reference-tracking problem by developing a distributed optimal coordinator, and the obtained augmented system is then stabilized by a high-gain observer based output feedback controller. It is proved that all the agent outputs converge to the optimal solution of the sum of local cost functions semi-globally. The effectiveness of the control scheme is illustrated by a simulation example. Applying the proposed distributed controller to practical systems, such as multi-robot systems and power systems, will be considered in our future research.

\begin{figure}[t]
	\begin{center}
		\setlength{\abovecaptionskip}{-0.1cm} 
		\includegraphics[height=4.2cm, width=8cm]{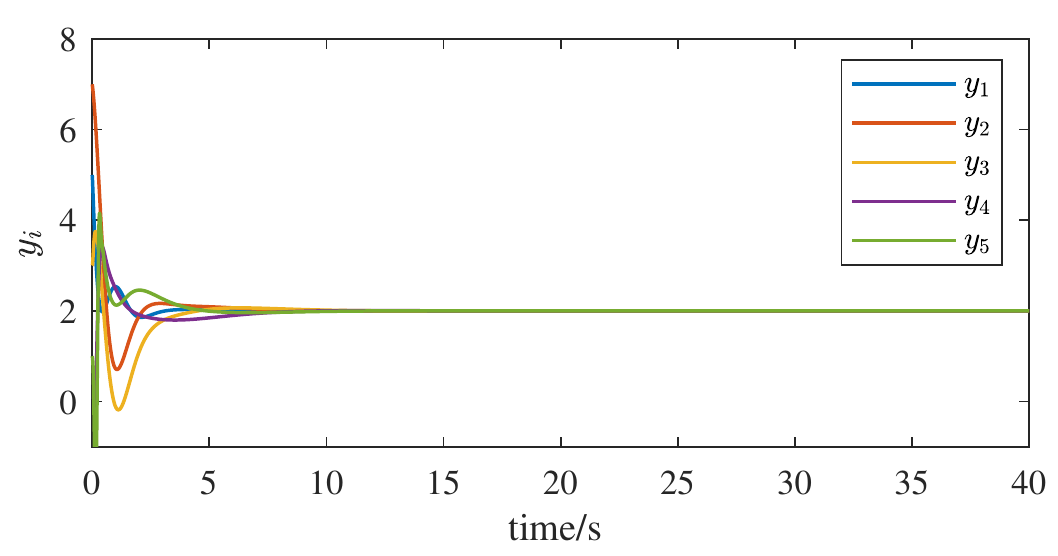} 
		\caption{Trajectories of the agent outputs $ y_{i}, i=1,\ldots,5 $.}  
		\label{Fig_second_order_example1_1}
		\vspace{-0.4cm}
	\end{center}
\end{figure}
%%%%%%%%%%%%%%%%%%%%%%%%%%%%%%%%%%%%%%%%%%%%%%%%%%%%%%%%%%%%%%%%%%%%%%%%%%%%%%%%%%%%%%%%%%%%%%%%%%%%%%%%%%%%%%%%%%%%%%%%%%%%%%%%%%%%%%%%%%%%%%%%%%%%%%%%%%%%%%

\hypertarget{appendices}{\appendices}

%\begin{appendices}
	
\section{Proof of Lemma \ref{Lemma_X}} \label{Appendix_1}
Note that $A_{c i}$ and $M_{i}$ are Hurwitz. Then, there exist positive definite matrices $P_{i 1}$ and $P_{i 2}$ such that $A_{c i}^{T} P_{i 1}+P_{i 1} A_{c i} \leq-I_{n_{i}-1}$, and $M_{i}^{T} P_{i 2}+P_{i 2} M_{i} \leq-I_{s_{i}}$. It is shown in Theorem \ref{proposition1} that $ (y_{i}^{r},\zeta_{i}) $ generated by (\ref{algorithm_r}) converges to its equilibrium $ (\bar{y}_{i}^{r}, \bar{\zeta}_{i}) $ exponentially. Thus, by applying Theorem 4.14 in \cite{khalil2002nonlinear}, there exists a positive definite Lyapunov function $ \tilde{V}_{i}(\tilde{y}_{i}^{r},\tilde{\zeta}_{i}) $ such that $ \dot{\tilde{V}}_{i}|_{(\ref{algorithm_r})} \leq-\varsigma_{i 0}\big(\left\|\tilde{y}_{i}^{r}\right\|^{2}+\|\tilde{\zeta}_{i}\|^{2}\big) $ for a constant $ \varsigma_{i 0}>0 $, where $ \tilde{y}_{i}^{r}= y_{i}^{r}-\bar{y}_{i}^{r} $ and $ \tilde{\zeta}_{i}= \zeta_{i}-\bar{\zeta}_{i} $. 

Consider the following positive definite functions, $V_{i 1}\left(\hat{x}_{i a}\right)=\hat{x}_{i a}^{T} P_{i 1} \hat{x}_{i a}$, $V_{i 2}\left(\tilde{\eta}_{i}\right)=\tilde{\eta}_{i}^{T} P_{i 2} \tilde{\eta}_{i}, V_{i 3}\left(\vartheta_{i}\right)=\frac{1}{2} \vartheta_{i}^{2}$. Then for any $\bar{R}>0$, there exists $\bar{c}_{i}>0$ such that $\bar{Q}_{\bar{R}}^{n_{z_{i}}+n_{i}+s_{i}+2} \subseteq \bar{\varOmega}_{\bar{c}_{i}}\left(V_{i 0}\left(\bar{z}_{i}\right)\right) \times \bar{\varOmega}_{\bar{c}_{i}}\left(V_{i 1}\left(\hat{x}_{i a}\right)\right) \times \bar{\varOmega}_{\bar{c}_{i}}\left(V_{i 2}\left(\tilde{\eta}_{i}\right)\right) \times \bar{\varOmega}_{\bar{c}_{i}}\left(V_{i 3}\left(\vartheta_{i}\right)\right) \times \bar{\varOmega}_{\bar{c}_{i}} \big(\tilde{V}_{i}(\tilde{y}_{i}^{r},\tilde{\zeta}_{i})\big) $. Let $ c=(4+\mu) \sum_{i=1}^{\mathbb{N}} \bar{c}_{i}$, where $ \mu $ is a positive constant to be determined. Note that $\tilde{f}_{i 0}\left(\bar{z}_{i}, \hat{x}_{i a}, y_{i}^{r}, v, w\right)-\tilde{f}_{i 0}\left(\bar{z}_{i}, 0, y_{i}^{r}, v, w\right)$ is sufficiently smooth and vanishes at $\hat{x}_{i a}=0$, and $\partial V_{i 0}\left(\bar{z}_{i}\right) / \partial \bar{z}_{i}$ is continuously differentiable with $\partial V_{i 0}(0) / \partial \bar{z}_{i}=0 $.  Then, by Lemma 2 in \cite{su2014cooperative}, for all $\bar{z}_{i} \in \bar{\varOmega}_{c+\epsilon}\left(V_{i 0}\left(\bar{z}_{i}\right)\right)$, $\hat{x}_{i a} \in \bar{\varOmega}_{c+\epsilon}\left(V_{i 1}\left(\hat{x}_{i a}\right)\right)$, $  y_{i}^{r} \in \bar{\varOmega}_{c+\epsilon}\big(\tilde{V}_{i}(\tilde{y}_{i}^{r},\tilde{\zeta}_{i})\big) $, and all $\operatorname{col}(v, w) \in \mathbb{V} \times \mathbb{W}$, the following inequalities are satisfied,
\begin{align}
	&\big\|\partial V_{i 0}(\bar{z}_{i})/\partial \bar{z}_{i}\big\| \leq \varsigma_{i 1}\| \bar{z}_{i} \|, \label{A_1} \\
	&\big\|\tilde{f}_{i 0}\left(\bar{z}_{i}, \hat{x}_{i a}, y_{i}^{r}, v, w\right)-\tilde{f}_{i 0}\left(\bar{z}_{i}, 0, y_{i}^{r}, v, w\right)\big\| \leq \varsigma_{i 2}\left\|\hat{x}_{i a}\right\|, \label{A_2} \\
	&\big\|\hat{f}_{i 0}\left(y_{i}^{r}, s^{\star}, v, w\right)\big\| \leq \varsigma_{i 3}\left\|\tilde{y}_{i}^{r}\right\|, \label{A_3} 
\end{align}
for some constants $ \varsigma_{i l}\geq 1, l = 1, 2, 3 $. Therefore, the derivatives of $ V_{i k}(\cdot), k=0,1,2,3 $ along the trajectories of (\ref{augmented_system_1}) can be described as follows,
\begin{align*}
	\dot{V}_{i 0}(\bar{z}_{i})\leq &  -\tfrac{a_{i 0}}{2}\|\bar{z}_{i}\|^{2}+ \tfrac{\varsigma_{i 1}^{2}\varsigma_{i 2}^{2}}{a_{i 0}}\|\hat{x}_{i a}\|^{2} + \tfrac{\varsigma_{i 1}^{2}\varsigma_{i 3}^{2}}{a_{i 0}}\|\tilde{y}_{i}^{r}\|^{2},\\[1mm]
	\dot{V}_{i 1}(\hat{x}_{i a})\leq & -(g-2)\|\hat{x}_{i a}\|^{2}+\|P_{i 1}\|^{2}\vartheta_{i}^{2} + \|P_{i 1}\|^{2}\|\dot{y}_{i}^{r}\|^{2}, \\[1mm]
	\dot{V}_{i 2}(\tilde{\eta}_{i})\leq & -\tfrac{1}{2}\|\tilde{\eta}_{i}\|^{2}+4\|P_{i2}\tilde{f}_{i 1}\|^{2} + 4\|P_{i2}b_{i}^{-1}(w)N_{i}\|^{2}\|\varepsilon_{i}\|^{2}, \\[1mm]
	\dot{V}_{i 3}(\vartheta_{i})\leq & -(b_{i}(w)K-1)\vartheta_{i}^{2} + \tfrac{1}{2}\|\tilde{f}_{i 2}\|^{2} + \tfrac{1}{2}\|\varepsilon_{i}\|^{2}.
\end{align*}

By recalling the definitions of $\tilde{f}_{i 1}(\cdot)$ and $\tilde{f}_{i 2}(\cdot)$, it follows from Lemma 7.8 in \cite{huang2004nonlinear} that, for all $\operatorname{col}(v, w) \in \mathbb{V} \times \mathbb{W}$,
\begin{align*}
	4\|P_{i 2} \tilde{f}_{i 1}(\cdot)\|^{2} \leq & \varsigma_{i 4}(g)\big(\rho_{i 1}\left(\bar{z}_{i}\right)\left\|\bar{z}_{i}\right\|^{2}\\
	&+\rho_{i 2}\left(\hat{x}_{i a}\right)\left\|\hat{x}_{i a}\right\|^{2}+\rho_{i 3}\left(\vartheta_{i}\right) \vartheta_{i}^{2}\big),\\
	\frac{1}{2}\|\tilde{f}_{i 2}(\cdot)\|^{2} \leq & \varsigma_{i 5}(g)\big(\rho_{i 4}\left(\bar{z}_{i}\right)\left\|\bar{z}_{i}\right\|^{2}\\ 
	&+\rho_{i 5}\left(\hat{x}_{i a}\right)\left\|\hat{x}_{i a}\right\|^{2}+\rho_{i 6}\left(\vartheta_{i}\right) \vartheta_{i}^{2}\big)+ \varsigma_{i 6}\left\|\tilde{\eta}_{i}\right\|^{2},
\end{align*}
for some smooth functions $\varsigma_{i l}(\cdot) \geq 1, l=4,5$, $\rho_{i k}(\cdot) \geq 1, k=1,2, \ldots, 6$, and a positive constant $\varsigma_{i 6} \geq 1$. Thus, for all $\bar{z}_{i} \in \bar{\varOmega}_{c+\epsilon}\left(V_{i 0}\left(\bar{z}_{i}\right)\right)$, $\hat{x}_{i a} \in \bar{\varOmega}_{c+\epsilon}\left(V_{i 1}\left(\hat{x}_{i a}\right)\right)$, and all $\operatorname{col}(v, w) \in \mathbb{V} \times \mathbb{W}$, one has 
\begin{align} \label{A_4}
	\max \big\{\rho_{i 1}\left(\bar{z}_{i}\right)\!+\!\tfrac{\rho_{i 4}\left(\bar{z}_{i}\right)}{4 \varsigma_{i 6}},~ \rho_{i 2}\left(\hat{x}_{i a}\right)\!+\!\tfrac{\rho_{i 5}\left(\hat{x}_{i a}\right)}{4 \varsigma_{i 6}}\big\} \leq \tfrac{a_{i 0} \varsigma_{i 7}}{4},
\end{align}
where $\varsigma_{i 7} \geq 1$ is a constant.
Furthermore, by letting $\varsigma_{i 8}(g)=\max \{\varsigma_{i 4}(g), \varsigma_{i 5}(g)\}$, we have
\begin{align} \label{A_5}
	\left\|P_{i 1}\right\|^{2} + \tfrac{\rho_{i 3}(\vartheta_{i})}{\varsigma_{i 7}} +\tfrac{\rho_{i 6}(\vartheta_{i})}{4 \varsigma_{i 6}\varsigma_{i 7}} \leq \varsigma_{i 9},
\end{align}
for all $\vartheta_{i} \in \bar{\varOmega}_{4(c+\epsilon) \varsigma_{i 6} \varsigma_{i 7} \varsigma_{i 8}(g)}\left(V_{i 3}(\vartheta_{i})\right)$ and all $\operatorname{col}(v, w) \in \mathbb{V} \times \mathbb{W}$, where $\varsigma_{i 9} \geq 1$ is a constant.

Consider the Lyapunov function candidate $ W_{X}(X)=\sum_{i=1}^{N} W_{i}\big(X_{i}\big) $, where $ W_{i}\big(X_{i}\big)= V_{i 0}(\bar{z}_{i})+V_{i 1}(\hat{x}_{i a})+V_{i 2}(\tilde{\eta}_{i}) /(\varsigma_{i 7} \varsigma_{i 8}(g)) + V_{i 3}(\vartheta_{i})/(4 \varsigma_{i 6} \varsigma_{i 7} \varsigma_{i 8}(g)) + \mu\tilde{V}_{i}(\tilde{y}_{i}^{r},\tilde{\zeta}_{i}) $. Since $ \varsigma_{i 7} \varsigma_{i 8}(g)\geq 1 $ and $ 4 \varsigma_{i 6} \varsigma_{i 7} \varsigma_{i 8}(g)\geq 1 $, it follows from $X_{i} \in \bar{\varOmega}_{\bar{c}_{i}}\left(V_{i 0}(\bar{z}_{i})\right) \times \bar{\varOmega}_{\bar{c}_{i}}\left(V_{i 1}(\hat{x}_{i a})\right) \times \bar{\varOmega}_{\bar{c}_{i}}\left(V_{i 2}(\tilde{\eta}_{i})\right) \times \bar{\varOmega}_{\bar{c}_{i}}\left(V_{i 3}(\vartheta_{i})\right) \times \bar{\varOmega}_{\bar{c}_{i}}\big(\tilde{V}_{i}(\tilde{y}_{i}^{r}, \tilde{\zeta}_{i})\big)$ that $X \in \bar{\varOmega}_{c}\big(W_{X}(X)\big)$, i.e., equation (\ref{varOmega_c_X}) is satisfied.

In addition, it is noted that $X \in \bar{\varOmega}_{c+\epsilon}\left(W_{X}(X)\right)$ implies $\bar{z}_{i} \in \bar{\varOmega}_{c+\epsilon}\big(V_{i 0}(\bar{z}_{i})\big)$, $\hat{x}_{i a} \in \bar{\varOmega}_{c+\epsilon}\big(V_{i 1}(\hat{x}_{i a})\big)$, $ \tilde{y}_{i}^{r} \in \bar{\varOmega}_{c+\epsilon}\big(\tilde{V}_{i}(\tilde{y}_{i}^{r}, \tilde{\zeta}_{i})\big) $ and $\vartheta_{i} \in \bar{\varOmega}_{4(c+\epsilon) \varsigma_{i 6} \varsigma_{i 7} \varsigma_{i 8}(g)}\big(V_{i 3}(\vartheta_{i})\big)$. Therefore, inequalities $(\ref{A_1})$--$(\ref{A_5})$ are still satisfied for all $X\in \bar{\varOmega}_{c+\epsilon}\big(W_{X}(X)\big)$. Moreover, it can be obtained that
\begin{align*}
	\Big(\tfrac{4\|b_{i}^{-1}(w)P_{i 2}N_{i}\|^{2}}{\varsigma_{i 7}\varsigma_{i 8}(g)} 
	+& \tfrac{1}{8\varsigma_{i 6} \varsigma_{i 7} \varsigma_{i 8}(g)}\Big)\|\varepsilon_{i}\|^{2} + \tfrac{\varsigma_{i 1}^{2}\varsigma_{i 2}^{2}}{a_{i 0}}\|\tilde{y}_{i}^{r}\|^{2} \notag\\
	+& \|P_{i 1}\|^{2}\|\dot{y}_{i}^{r}\|^{2} \leq \varsigma_{i 10}\|\tilde{y}_{i}^{r}\|^{2} + \varsigma_{i 11}\|\dot{y}_{i}^{r}\|^{2},
\end{align*}
for some constants $\varsigma_{i l} \geq 1, l=10, 11$.
Thus, for all $X \in \bar{\varOmega}_{c+\epsilon}\left(W_{X}(X)\right)$ and all $\operatorname{col}(v, w) \in \mathbb{V} \times \mathbb{W}$, the derivative of $W_{i}(\cdot)$ along the trajectories of the closed-loop system composed of (\ref{algorithm_r}), (\ref{augmented_system_1}) and (\ref{state feedback}) can be described as follows,
\begin{align*}
	\dot{W}_{i}\big(\bar{z}_{i}, \hat{x}_{i a}, &\tilde{\eta}_{i}, \vartheta_{i}, \tilde{y}_{i}^{r}, \tilde{\zeta}_{i}\big)\big|_{(\ref{algorithm_r})+(\ref{augmented_system_1})+(\ref{state feedback})} \notag\\
	\leq &-\tfrac{a_{i 0}}{4}\left\|\bar{z}_{i}\right\|^{2}-\Big(g-2-\tfrac{\varsigma_{i 1}^{2} \varsigma_{i 2}^{2}}{a_{i 0}}-\tfrac{a_{i 0}}{4}\Big)\|\hat{x}_{i a}\|^{2} \notag\\
	&-\tfrac{1}{4 \varsigma_{i 7} \varsigma_{i 8}(g)}\left\|\tilde{\eta}_{i}\right\|^{2} -\Big(\tfrac{b_{i}(w)K - 1}{4 \varsigma_{i 6} \varsigma_{i 7} \varsigma_{i 8}(g)}- \varsigma_{i 9}\Big) \vartheta_{i}^{2} \notag\\[1mm]
	&+\varsigma_{i 10}\|\tilde{y}_{i}^{r}\|^{2} + \varsigma_{i 11}\|\dot{y}_{i}^{r}\|^{2} - \mu \varsigma_{i 0}\big(\left\|\tilde{y}_{i}^{r}\right\|^{2}+\|\tilde{\zeta}_{i}\|^{2}\big).
\end{align*}
Under Assumption \ref{assumption_cost functions}, it can be proved that $ \dot{y}_{i}^{r} $ is Lipschitz in $ \operatorname{col}(\tilde{y}_{i}^{r}, \tilde{z}_{i}) $ by using similar arguments as those in \cite{tang2020optimal}. One thus has $ \varsigma_{i 11}\|\dot{y}_{i}^{r}\|^{2} \leq \mu_{1}\big(\left\|\tilde{y}_{i}^{r}\right\|^{2}+\|\tilde{\zeta}_{i}\|^{2}\big) $,
%\begin{align} \label{hat_gamma_ir}
%	\varsigma_{i 11}\|\dot{y}_{i}^{r}\|^{2} \leq \mu_{1}\big(\left\|\tilde{y}_{i}^{r}\right\|^{2}+\|\tilde{\zeta}_{i}\|^{2}\big),
%\end{align}
where $ \mu_{1}>0 $ is a constant.
Then, by successively choosing $g>\max\limits_{i=1,2, \ldots, N}\big\{2+\tfrac{\varsigma_{i 1}^{2} \varsigma_{i 2}^{2}}{a_{i 0}}+\tfrac{a_{i 0}}{4}\big\}$, $ K> \max\limits_{i=1,2, \ldots, N}\big\{\frac{1}{\min_{w \in \mathbb{W}}\{b_{i}(w)\}}(1+4 \varsigma_{i 6} \varsigma_{i 7} \varsigma_{i 8}(g) \varsigma_{i 9})\big\}$, $ \mu>\max\limits_{i=1,2, \ldots, N}\big\{\frac{\varsigma_{i 10}+\mu_{1}}{\varsigma_{i 0}}\big\} $, and $ a=\min\limits_{i=1,2, \ldots, N}\big\{\tfrac{a_{i 0}}{4},~ g-2-\tfrac{\varsigma_{i 1}^{2} \varsigma_{i 2}^{2}}{a_{i 0}}-\tfrac{a_{i 0}}{4},~ \tfrac{1}{4 \varsigma_{i 7} \varsigma_{i 8}(g)},~ \tfrac{b_{i}(w)K - 1}{4 \varsigma_{i 6} \varsigma_{i 7} \varsigma_{i 8}(g)}- \varsigma_{i 9},~ \mu \varsigma_{i 0}- \mu_{1}- \varsigma_{i 10}\big\} $, one can obtain that $\dot{W}_{i}\left(X_{i}\right)\big|_{(\ref{algorithm_r})+(\ref{augmented_system_1})+(\ref{state feedback})} \leq -a\left\|X_{i}\right\|^{2}$. This indicates that inequality (\ref{dot_W_X}) is satisfied, and the proof is thus completed. \qed

\section{Proof of Lemma \ref{Lemma_e}} \label{Appendix_2}
Under appropriate choices of $ c_{i k}, k=1,2,\ldots,n_{i} $, matrix $ E_{i 3} $ defined in (\ref{E_i3}) is Hurwitz. Thus, there exists a positive definite matrix $P_{i 3}$ such that $E_{i 3}^{T} P_{i 3}+P_{i 3} E_{i 3} \leq -I_{n_{i}}$. Consider the positive definite function $ W_{e}(e)=\sum_{i=1}^{N}e_{i}^{\operatorname{T}}P_{i 3}e_{i} $. Then, for all $\operatorname{col}(v, w) \in$ $\mathbb{V} \times \mathbb{W}$, its derivative along the trajectories of (\ref{e_{i}_1}) can be deduced as follows,
\begin{equation*}
	\dot{W}_{e}(e)\big|_{(\ref{e_{i}_1})} \leq -\frac{h}{2}\sum_{i=1}^{N}\|e_{i}\|^{2} + \frac{2}{h} \sum_{i=1}^{N} \|P_{i 3}E_{i 4}\|^{2}\|\dot{X}_{i}\|^{2}.
\end{equation*}
Let $ \varrho= \max\limits_{i=1,2, \ldots, N}\big\{2\|P_{i 3}E_{i 4}\|^{2}\big\} $. Then the inequality (\ref{dot_W_e}) is satisfied, and the proof is thus completed. \qed

%\end{appendices}

{
\tiny
\bibliographystyle{IEEEtran}  %这是你要使用的格式,比如要投IEEE,就写IEEEtran
\bibliography{IEEEabrv,mylib}  %这个是加载你的bib,你可以理解从文献数据库中加载要引用的文献
}

\end{document}